\documentclass[a4paper,UKenglish,cleveref,autoref,thm-restate,authorcolumns]{lipics-v2021}
\nolinenumbers
\hideLIPIcs
\usepackage{microtype}
\usepackage{graphicx} 
\usepackage{todonotes}
\usepackage{comment}
\usepackage{amsthm} 
\usepackage{amsfonts} 
\usepackage{amssymb} 
\usepackage{amsmath} 
\usepackage{mathtools} 
\usepackage{booktabs}
\usepackage{hyperref}
\usepackage{tabularx}
\usepackage{cleveref}
\usepackage{enumitem}
\usepackage{etoolbox}
\AtBeginEnvironment{alignat}{\setcounter{equation}{0}} 

\title{Logical Approaches to Non-deterministic Polynomial Time over Semirings}
\author{Timon {Barlag}}{Leibniz University Hannover, Germany
\and  }{barlag@thi.uni-hannover.de}{https://orcid.org/0000-0001-6139-5219}{}
\author{Nicolas {Fröhlich}}{Leibniz University Hannover, Germany}{nicolas.froehlich@thi.uni-hannover.de}{https://orcid.org/0009-0003-5413-1823}{Funded by the German research foundation (DFG), project ME 4279/3-1 (project id 511769688)}
\author{Teemu {Hankala}}{University of Helsinki, Finland  }{teemu.hankala@helsinki.fi}{https://orcid.org/0009-0009-5535-5718}{}
\author{Miika {Hannula}}{University of Tartu, Estonia \and University of Helsinki, Finland }{miika.hannula@ut.ee}{https://orcid.org/0000-0002-9637-6664}{Partially supported by the ERC
grant 101020762.}
\author{Minna {Hirvonen}}{Leibniz University Hannover, Germany
 \and University of Helsinki, Finland }{minna.hirvonen@thi.uni-hannover.de}{https://orcid.org/0000-0002-2701-9620}{Funded by the Magnus Ehrnrooth foundation.}
\author{Vivian {Holzapfel}}{Leibniz University Hannover, Germany
 }{holzapfel@thi.uni-hannover.de}{https://orcid.org/0009-0001-1439-5037}{}
\author{Juha {Kontinen}}{University of Helsinki, Finland }{juha.kontinen@helsinki.fi}{https://orcid.org/0000-0003-0115-5154}{}
\author{Arne {Meier}}{Leibniz University Hannover, Germany
}{meier@thi.uni-hannover.de}{https://orcid.org/0000-0002-8061-5376}{Partially supported by the German research foundation (DFG), project ME 4279/3-1 (project id 511769688)}
\author{Laura {Strieker}}{Leibniz University Hannover, Germany
}{strieker@thi.uni-hannover.de}{https://orcid.org/0009-0005-4878-4953}{}

\authorrunning{T. Barlag et al.}
\funding{Finnish-German co-operation grant: Research Council of Finland (359650), German Academic Exchange Service DAAD (57710940)}

\newcommand{\lit}{\mathrm{Lit}}
\newcommand{\logicFont}[1]{\mathrm{#1}}
\newcommand{\FOK}{{\logicFont{FO}}_{K}^{=}}
\newcommand{\FO}{{\logicFont{FO}}}
\newcommand{\AP}{\mathrm{Prop}}
\newcommand{\PLK}{\mathrm{PL}_{K}^{=}}
\newcommand{\ESOK}{\mathsf{ESO}_K^{=}}

\newcommand{\bnot}{\lnot_{\mathbb{B}}}
\newcommand{\bor}{\lor_{\mathbb{B}}}
\newcommand{\band}{\land_{\mathbb{B}}}
\newcommand{\bto}{\to_{\mathbb{B}}}

\newcommand{\R}{\mathbb{R}}

\newcommand{\N}{\mathbb{N}}
\newcommand{\Z}{\mathbb{Z}}

\newcommand{\ar}{\textnormal{ar}}

\newcommand{\ddfn}{\Coloneqq}
\newcommand{\dfn}{\coloneqq}
\newcommand{\leqm}{\leq_{\rm m}}
\newcommand{\leqpm}{\leqm^{\PTIME}}

\newcommand{\BSSK}{\mathrm{BSS}_K}

\usepackage{booktabs}

\newcommand{\fv}{\mathrm{FV}}
\newcommand{\Var}{\mathrm{Var}}

\newcommand{\MC}{\text-\mathrm{MC}}
\newcommand{\FOKMC}[2][O]{\FOK(#1)\MC\def\temp{#2}\ifx\temp\empty\else_#2\fi}
\newcommand{\SATK}{\mathrm{SAT}_{K}^{=}}
\newcommand{\SATN}{\mathrm{SAT}_{\N}^{=}}
\newcommand{\SATKflat}{\mathrm{SAT}_{K}^\mathrm{flat}}

\usepackage{stmaryrd}
\newcommand{\evaluate}[2]{\llbracket#1\rrbracket_{#2}}

\usepackage{tikz}
\usetikzlibrary{
    arrows,
    positioning,
}

\usepackage[linesnumbered, ruled, vlined]{algorithm2e}

\newcommand{\PTIME}{\mathsf{P}}
\newcommand{\PTIMEK}{\mathsf{P}_K}
\newcommand{\NP}{\mathsf{NP}}
\newcommand{\NPK}{\mathsf{NP}_K}

\newcommand{\NPKX}[1]{\NPK[#1]}
\newcommand{\NPKXX}{\NPKX{X}}
\newcommand{\NPR}{\mathsf{NP}_{\mathbb{R}}}
\newcommand{\ETK}{\mathsf{ETK}}
\newcommand{\ETKX}[1]{\ETK[#1]}
\newcommand{\ETKXK}{\ETKX{K}}
\newcommand{\ETKXX}{\ETKX{X}}
\newcommand{\EK}{\exists K}

\newcommand{\ETR}{\mathsf{ETR}}
\newcommand{\leqpmk}{\leqm^{\PTIMEK}}
\newcommand{\leqpmkx}[1]{\leqm^{\PTIMEK[#1]}}
\newcommand{\leqpmkxx}{\leqpmkx{X}}
\newcommand{\BP}[1]{\operatorname{BP}(#1)}

\newcommand{\ourpar}[1]{%
  \par\vspace{.25\baselineskip}%
  \noindent\textbf{\sffamily #1.}%
}

\begin{document}

\maketitle

\begin{abstract}
 We provide a logical characterization  of non-deterministic polynomial time defined by BSS machines over semirings via existential second-order logic interpreted in the semiring semantics developed by Gr\"adel and Tannen. 
Furthermore, we show that, similarly to the classical setting, the satisfiability problem of propositional logic in the semiring semantics is the canonical complete problem for this version of NP. 
Eventually, we prove that the true existential first-order theory of the semiring is a complete problem for the so-called Boolean part of this version of NP.

\ccsdesc[500]{Theory of computation~Abstract machines}
\ccsdesc[500]{Theory of computation~Turing machines}
\ccsdesc[500]{Theory of computation~Complexity classes}
\ccsdesc[300]{Theory of computation~Problems, reductions and completeness}
\ccsdesc[300]{Theory of computation~Complexity theory and logic}

\keywords{Semiring,
ESO, SAT, BSS machines, computational complexity, descriptive complexity} 
\end{abstract}

\section{Introduction}

Since the 1960s, many areas of computer science have undergone algebraic generalizations, driving advances in algorithms, automata theory, optimization, and machine learning, while also shedding new light on the classical complexity of fundamental problems such as matrix multiplication \cite{MR248973,S0097539793243004,
EiterK23,
10.5555/639508.639512}. 
In particular, \emph{semirings} have found numerous applications in computer science and AI due to their versatility and modularity in modelling and analysing computational problems
(see, e.g., \cite{DBLP:journals/fuin/RudeanuV04,DBLP:journals/ijac/GaubertK06}).
Classical decision problems are typically framed within the Boolean semiring ($\mathbb{B}$). Counting problems, by contrast, are naturally modelled in the semiring of natural numbers ($\mathbb{N}$), whereas problems involving continuous or geometric aspects are captured by the semiring of non-negative reals ($\mathbb{R}_{\geq 0}$). 
More recently, research combining logic and algebraic semantics has begun to emerge, e.g.,  \cite{GM95,DROSTE200769,10.1145/3643027,GradelM21,HannulaKBV20,BarlagHKPV23}.

Several computational models can be generalized to operate over semirings. For instance, weighted automata and weighted Turing machines assign semiring elements to transitions, representing quantities such as probabilities, costs, or capacities. By varying the underlying semiring, these models capture a wide range of quantitative behaviours, with applications in probabilistic reasoning, optimization, and verification (e.g., \cite{KOSTOLANY,EiterK23}). The objects recognized by such machines are formal power series $r\colon \Sigma^* \to K$, defined over a finite alphabet $\Sigma$ and a semiring $K$.
Conversely, languages over the reals---and more generally over rings---have been investigated through register machines, which operate on real numbers instead of Booleans and perform real-valued functions in a single computational step. 
The well-known Blum-Shub-Smale (BSS) machine model, along with the closely related real RAM model, originates in \cite{DBLP:conf/focs/BlumSS88,PreparataS85}. 
Moreover, the BSS model extends naturally to semirings and arbitrary first-order structures \cite{Poi95,IC2005,Jelia25}.

BSS machines can be also used to  decide Boolean languages. Indeed, the \emph{Boolean part} of NP over the reals ($\NPR$) (i.\,e., with machine constants $0$ and $1$, and inputs restricted to $\{0,1\}^*$) coincides with the class $\exists \mathbb{R}$, consisting of all decision problems that are polynomial-time reducible to the existential first-order theory of the reals--equivalently, to the non-emptiness of semialgebraic sets \cite{SchaeferS17}. The class $\exists \mathbb{R}$ has attracted considerable attention in recent years, as it captures the complexity of a wide range of geometric problems (see the recent survey \cite{2407-18006}). 
 The classes $\NPR$ and $\exists \mathbb{R}$ have been logically characterized by existential second-order logic in the metafinite model theory context  \cite{GM95}.

During the last decade, the use of semirings to study \emph{provenance} for database query languages has attracted increasing amounts of attention~\cite{%
GreenKT07}.
 Semiring provenance is an approach to query evaluation in which the result of a query is something more than just a mere one-bit true/false answer. 
 The basic idea behind this approach is to annotate the atomic facts in  database tables by values from some semiring $K$ ($K$-relation), and to propagate these values through a query.
 Depending on the choice of the semiring, the provenance valuation gives information about a query, e.g., regarding its confidence, cost, or the number of assignments that make the query true.
 Semiring semantics for query languages is currently an actively studied topic in database theory (see, e.g.,
  \cite{10.1145/3651146,im_et_al:LIPIcs.ICDT.2024.11} for Datalog queries and  \cite{eldar_et_al:LIPIcs.ICDT.2024.4,munozserrano_et_al:LIPIcs.ICDT.2024.12,tencate_et_al:LIPIcs.ICDT.2024.8} for conjunctive queries). It is worth noting that $K$-relations (or weighted structures) provide a unifying abstraction that subsumes, e.g.,  probability distributions and bags, with possible applications extending beyond database provenance (see, e.g., \cite{AtseriasK25,grohe_et_al:LIPIcs.ICDT.2025.9}).
  
 Semiring semantics has also been defined for first-order logic \cite{abs-1712-01980, Grädel2025,Tannen17}. In this context, it has been investigated in relation to topics including Ehrenfeucht–Fraïssé games, locality, 0–1 laws, and definability up to isomorphism \cite{BiziereGN23,BrinkeGM24,GradelHNW22,GradelM21}.
 Semiring semantics has likewise been defined for more expressive logics, such as fixed-point logic \cite{DannertGNT21} and team-based logics~\cite{BarlagHKPV23}.   
 Recently,  the expressivity and model checking complexity of extensions of first-order logic under the semiring semantics  have been characterized by arithmetic circuits and BSS computations over $K$ \cite{Jelia25}. 
 
 \ourpar{Contributions}
 In this article, we extend the study of Barleg~et~al.~\cite{Jelia25} to existential second-order logic in the semiring semantics and prove analogues of Fagin's and Cook's theorems. We also generalize the aforementioned connection between the Boolean part of $\NPR$ and $\ETR$ from the reals to positive and commutative semirings. It is worth noting that our work bears analogy, e.g., to recent articles by Badia et al.~\cite{DBLP:journals/corr/abs-2507-18375,badia_et_al:LIPIcs.MFCS.2024.14} in which versions of Fagin's theorem are established for different models of weighted TMs and to work by Grädel and Meer~\cite{GM95} which provides a logical characterization for $\NPR$ in the metafinite logic framework. Our results imply that, for expressive power, $\ESOK$ and the metafinite version of Grädel and Meer are intimately related and warrant further study.

\ourpar{Organization}
In Section \ref{Prel}, we discuss the basics of semirings and our background assumptions.
In Section~\ref{logic}, we define semirings semantics for propositional ($\PLK$), first-order ($\FOK$) and existential second-order logic ($\ESOK$). The section also includes the definition of the satisfiability problem of $\PLK$. 
In Section~\ref{bss_section}, we discuss the  BSS model over semirings and in Section \ref{Fagin} we show a version of Fagin's theorem for $\ESOK$ and $\NPK$ that holds for positive and commutative semirings. In Section \ref{Cook}, we turn to the analogue of Cook's theorem and in Section \ref{ETKand} we show that the true existential theory of $K$ is a complete problem for the  Boolean part of $\NPK$. Along the way, we introduce a novel non-arithmetic reduction that allows us to generalize the aforementioned result considerably. The paper ends with a conclusion.


\section{Preliminaries}\label{Prel}
We assume familiarity with basic concepts in theoretical computer science, e.g., complexity theory~\cite{DBLP:books/daglib/0086373}.
We start with the fundamental definition of a semiring.
\begin{definition}
    A \emph{semiring} is a tuple $K=(K,+,\cdot,0,1)$, where $+$ and $\cdot$ are binary operations on $K$, $(K,+,0)$ is a commutative monoid with identity element $0$, $(K,\cdot ,1)$ is a monoid with identity element $1$, $\cdot$ left and right distributes over $+$, and $x \cdot 0 =0= 0\cdot x$ for all $x \in K$.
    $K$ is called \emph{commutative} if $(K,\cdot ,1)$ is a commutative monoid.
    As usual, we often write $ab$ instead of $a \cdot b$. 
    A semiring is \emph{ordered} if there exists a  partial order $\leq$ (meaning $\leq$ is reflexive, transitive, and antisymmetric) such that $0\le 1$, and for all $a,b, c \in K$ $a\leq b$ implies $a+c \leq b+c$, and $a\leq b \land 0 \leq c$ implies $ac \leq bc \land ca \leq cb$.
    A semiring is \emph{positive} if it has no divisors of $0$, i.\,e., $ab\neq 0$ for all $a,b\in K$, where $a\neq 0 \neq b$ and if $a+b=0$ implies that $a=b=0$.
\end{definition}
Throughout the paper let $K$ be a semiring.
The \emph{natural order} $\leq_{\rm n}$ of $K$ is  defined by $a\leq_{\rm n} b$ iff $\exists c: a+c=b$. 
 The natural order of a semiring is always a preorder (meaning $\leq$ is reflexive and transitive). Moreover, if $\leq_{\rm n}$ is antisymmetric, then the semiring is ordered by it.

Important examples of semirings are the \emph{Boolean semiring} $\mathbb{B}=(\mathbb{B},\lor,\land,0,1)$ as the simplest example of a semiring that is not a ring, the \emph{semiring of non-negative reals} $\mathbb{R}_{\geq 0}=(\mathbb{R}_{\geq 0},+,\cdot,0,1)$,
and the \emph{semiring of natural numbers} $\mathbb{N}=(\mathbb{N},+,\cdot,0,1)$. 
Further examples include
the \emph{tropical semiring} $\mathbb{T} = (\mathbb{R}\cup\{\infty\}, \min, +, \infty, 0)$,
the \emph{semiring of multivariate polynomials} $\mathbb{N}[X]=(\mathbb{N}[X],+,\cdot,0,1)$,
and the \emph{\L{}ukasiewicz semiring} $\mathbb{L} = ([0,1], \max, \cdot , 0, 1)$, used in multivalued logic.

\begin{remark}
Throughout this paper, we consider semirings that are commutative, positive, and non-trivial (i.\,e., $0 \neq 1$). 
These assumptions are made for the sake of simplicity, even though the results in Section~\ref{ETKand} do not require these assumptions except for non-triviality.
\end{remark}

\section{Logics over Semirings}\label{logic} 
In this section, 
we define the semiring semantics for propositional, first-order and existential second-order formulae 
with so-called formula (in)equality atoms~\cite{abs-1712-01980, Grädel2025,BarlagHKPV23}. 
These logics will also have access to arbitrary constants from $K$.
We will introduce them using a superscript $=$ in their names to differentiate them from logics the one defined by Grädel and Tannen~\cite{abs-1712-01980, Grädel2025}, and to emphasize that they have access to said formula (in)equalities.

\subsection{Propositional Logic over Semirings and the Satisfiability Problem}

We first introduce $\PLK$, the propositional logic over semirings.
Let $\AP$ be a countably infinite set of propositions. 
An assignment $s$ is a mapping $s\colon\{\,p,\lnot p\mid p\in\AP\,\}\to K$ that associates each atomic proposition with a semiring value.
 \begin{definition}[$\PLK$]Let $K = (K,+,\cdot,0,1)$ be a semiring and $\AP$ a set of propositional symbols.
The \emph{syntax of propositional logic over $K$ and $\AP$} is defined as follows:
    \[
          \phi\ddfn p \mid \neg p \mid c\mid \phi = \phi \mid (\phi \wedge \phi) \mid (\phi \lor \phi),
    \]
    where $p\in \AP$, $c\in K$. For ordered $K$ we additionally have $\phi \leq \phi$.
    
 Let $s\colon\{\,p,\neg p\mid p\in\AP\,\}\to K$ be an assignment. The \emph{sematics of a $\PLK$-formula $\theta$ under an assignment $s$} is denoted by $\evaluate{\theta}{s}$ and defined as follows: 
\begin{align*}
        \evaluate{p_i}{s} &= s(p), &
        \evaluate{\neg p}{s} &= s(\neg p),
        &
        \evaluate{c}{s} &= c,\\
        \evaluate{\phi\wedge\psi}{s} &= \evaluate{\phi}{s}\cdot\evaluate{\psi}{s}, &
        \evaluate{\phi\lor\psi}{s} &= \evaluate{\phi}{s}+\evaluate{\psi}{s}, \\
        \evaluate{\phi\circ\psi}{s} &=
        \begin{cases}
        1 \quad \text{if } \evaluate{\phi}{s}\circ\evaluate{\psi}{s},\\
        0 \quad \text{otherwise}, 
        \end{cases}
        \text{where } \circ\in\{=,\leq\}. \hspace{-8cm}
    \end{align*}
\end{definition}
We will make use of shortcut notations that allow us to use any of the following comparison operators in $\PLK$ formulae:
$   \phi \neq \psi \coloneqq (\phi = \psi) = 0,
    \phi \nleq \psi \coloneqq (\phi \leq \psi) \leq 0.$

Next, we define the satisfiability problem $\SATK$ for a semiring $K$ in the context of this logic. 
Since there is not exactly one value denoting truth in general, we define satisfaction with respect to any non-zero semiring value.
\begin{definition} \label{def:satk}
    The \emph{satisfiability problem} for $\PLK$ is defined as 
    \[
        \SATK \coloneqq \{\,\varphi \in \PLK \mid \text{ there exists an assignment $s$ such that $\evaluate{\varphi}{s} \neq 0$}\,\}.
    \]
\end{definition}
\begin{example}
Consider the semiring $\N$ of natural numbers.
The formula 
    $\phi = ((p \lor q) \leq (p \land q)) \land 3 \leq p \land 3 \leq q$
is satisfied by the assignment $s(p) = s(q) = 3$, with
$\evaluate{\phi}{s} = ((3 + 3) \leq (3 \cdot 3)) \cdot (3 \leq 3) \cdot (3 \leq 3) = 1$. 
As a result, we have that $\phi \in \SATN$.
On the other hand, the slightly modified formula
    $\psi = ((p \land q) \leq (p \lor q)) \land 3 \leq p \land 3 \leq q$
is not satisfied by any assignment, because either $p < 3$, $q < 3$ or $p \cdot q > p + q$.
Therefore, we get that $\evaluate{\psi}{s} = 0$ and $\psi \not\in \SATN$.
\end{example}
\subsection{First-order Logic over Semirings}

In the following, we denote by $\ar(R)$ the \textit{arity} of a relational symbol $R$. 
The set $\Var$ is the set of first-order variables.
A vocabulary $\tau$ is a set of relation symbols. 

\begin{definition}[$\FOK$] 
    \label{def:fok}
    The \emph{syntax for first-order logic over a semiring $K$} (of vocabulary~$\tau$), is defined as follows:
     \[
    \phi\ddfn x = y \mid  x \neq y \mid c\mid R(\bar{x})\mid \neg R(\bar{x}) \mid \phi = \phi \mid(\phi\wedge\phi) \mid (\phi\lor\phi) \mid \exists x\phi \mid \forall x\phi,
    \]
    where $x,y\in\Var$, $c\in K$, $R\in\tau$ and $\bar{x}$ is a tuple of variables so that $\ar(R)=|\bar{x}|$.
    For ordered semirings $K$, we additionally have that $\phi \leq \phi$. 
\end{definition}

The definition of a set of \emph{free variables} $\fv(\phi)$ of a formula $\phi \in \FOK$ extends from $\FO$ in the obvious way: $\fv(\phi\circ \psi)=\fv(\phi)\cup\fv(\psi)$, for $\circ\in \{=,\leq\}$. If $\phi$ contains no free variables, i.\,e., $\fv(\phi)=\emptyset$, it is called a \emph{sentence}. 

Let $A$ be a finite set (i.\,e., the model domain), $\tau=\{R_1,\dots,R_n\}$ a finite vocabulary, and $K$ a semiring. 
Define $\lit_{A,\tau}$ as the set of \emph{literals} over $A$, i.\,e., the set of facts and negated facts $\lit_{A,\tau}=\{\,R(\bar{a}),\neg R(\bar{a})\mid \bar{a}\in A^{\ar(R)},R\in\tau\,\}$. 
A \emph{$K$-interpretation} $\pi$ is a mapping $\pi\colon \lit_{A,\tau}\to K$. 
A $K$-interpretation $\pi$ is \emph{ordered}, if $A = \{a_0, a_1, \dots, a_{\lvert A \rvert - 1}\}$ and $\tau$ contains a binary relation $R_<$ such that 
\[
    \pi(R_<(a_i, a_j)) = 
    \begin{cases}
        1, & \textnormal{if $i < j$}, \\
        0, & \textnormal{otherwise.}
    \end{cases}
\]
An \emph{assignment} $s$ is a mapping $s\colon\Var\to A$ that associates each first-order variable with an element from the set $A$.
     
\begin{definition}\label{interpretation}
    Let $K=(K,+,\cdot,0,1)$ be a semiring, $\pi$ be a $K$-interpretation, and $s$ be an assignment.
    The \emph{semantics of a $\FOK$-formula
    $\theta$ under $\pi$ and $s$}, denoted by $\evaluate{\theta}{\pi,s}$ and defined as follows: 
    \begin{align*}
        \evaluate{x\star y}{\pi,s} &=
        \begin{cases}
        1 \quad \text{if } s(x)\star s(y),\\
        0 \quad \text{otherwise}, 
        \end{cases}
        \text{where } \star\in\{=,\neq\}, \hspace{0cm}
        &\evaluate{c}{\pi,s}&=c\\
        \evaluate{R(\bar{x})}{\pi,s} &= \pi(R(s(\bar{x}))), &
        \evaluate{\neg R(\bar{x})}{\pi,s} &= \pi(\neg R(s(\bar{x}))),\\
        \evaluate{\phi\wedge\psi}{\pi,s} &= \evaluate{\phi}{\pi,s}\cdot\evaluate{\psi}{\pi,s}, &
        \evaluate{\phi\lor\psi}{\pi,s} &= \evaluate{\phi}{\pi,s}+\evaluate{\psi}{\pi,s}, \\
        \evaluate{\exists x\phi}{\pi,s} &= \sum_{a\in A}\evaluate{\phi}{\pi,s(a/x)}, & \evaluate{\forall x\phi}{\pi,s} &= \prod_{a\in A}\evaluate{\phi}{\pi,s(a/x)}, \\
        \evaluate{\phi\circ\psi}{\pi,s} &=
        \begin{cases}
        1 \quad \text{if } \evaluate{\phi}{\pi,s}\circ\evaluate{\psi}{\pi,s},\\
        0 \quad \text{otherwise}, 
        \end{cases}
        \text{where } \circ\in\{=,\leq\}. \hspace{-4cm}
    \end{align*}
    If $\phi$ is a \emph{sentence}, we write $\evaluate{\phi}{\pi}$ for $\evaluate{\phi}{\pi,\emptyset}$.
\end{definition}
We can again define $\neq$ and $\nleq$ for formulas, in the same way as was done for $\PLK$.
Sometimes, $K$-interpretations $\pi$ are called \textit{model-defining} \cite{abs-1712-01980, Grädel2025}, which is the case if, for all $R(\bar{{a}})$, we have that $\pi(R(\bar{a}))=0$ if and only if $\pi(\neg R(\bar{a}))\neq 0$.
Sentences $\phi$ and $\psi$ of $\FOK$ are \emph{$K$-equivalent}, if $\evaluate{\phi}{\pi}=\evaluate{\psi}{\pi}$ is true for all 
$K$-interpretations. 
 
The following definition provides a way to simulate the classical Boolean connectives in the semiring framework.
\begin{definition}[Boolean connectives]\label{def:bool}
    The classical Boolean connectives are defined as follows:
    \begin{align*}
        \bnot \phi &\coloneq (\phi = 0) \\
        \phi \bor \psi &\coloneq (\phi \neq 0 \lor \psi \neq 0) \neq 0 \\
        \phi \band \psi &\coloneq (\phi \neq 0 \land \psi \neq 0) \neq 0 \equiv \bnot(\bnot \phi \bor \bnot \psi) \\
        \phi \bto \psi &\coloneq (\phi = 0 \lor \psi \neq 0) \neq 0 \equiv (\bnot \phi \bor \psi) 
    \end{align*}
\end{definition}

\subsection{Existential Second-order Logic over Semirings}
We define existential second-order quantification such that it corresponds to the existence of a $K$-relation in the Boolean sense, i.\,e., a formula with second-order quantification can only have values 0 or 1. Another natural option would be to define existential second-order quantification analogously to existential first-order quantification, by using a sum over all possible $K$-interpretations of the quantified relation. Note that in our setting, we are mainly interested in infinite semirings, since the operation table of a finite semiring can be handled with a standard Turing machine. This means that if the semantics of second-order quantification is defined by using sums, we would need to consider infinite sums.

\begin{definition}[$\ESOK$]

    The \emph{syntax for existential second-order logic} (of vocabulary $\tau$) over a (ordered) semiring $K$, is defined as follows
    \[
        \phi\ddfn \alpha \mid \exists R \phi,
    \]
    where $\alpha\in\FOK$ and $R$ is a relation symbol not in $\tau$. 
    We say that a $K$-interpretation $\pi'\colon \lit_{A,\tau\cup\{R\}}\to K$ \emph{extends} $\pi$ \emph{with} $R$ if $\pi'(\ell)=\pi(\ell) \text{ for all } \ell\in \lit_{A,\tau}$. 
    The \emph{semantics of existential second-order quantification} is then defined as follows.
    \[
        \evaluate{\exists R\phi}{\pi,s}:=
        \begin{cases}
            1, \text{ if there is } \pi' \text{ that extends } \pi \text{ with } R\text{ such that }\evaluate{\phi}{\pi',s}\neq 0, \\
            0, \text{ otherwise}.
        \end{cases}
    \]
\end{definition}
This definition enforces that all quantified relation symbols are in the beginning of the formula. 
Allowing quantified relations to appear anywhere within a formula, especially inside comparisons, would allow expressive power beyond existential second-order logic.
This is the case, because it would allow that negated existential second-order quantifier could be constructed, e.\,g., $\lnot \exists P \psi$  could be simulated by $(\exists P \psi)=0$.

In Example~\ref{ex:4FEAS} below, we show that the set $4\textnormal{-FEAS}_{\geq 0}$ of degree $4$ polynomials over $\R_{\geq 0}$ that have a zero in $\R_{\geq 0}$ can be characterized in $\mathrm{ESO}_{\R_{\geq 0}}^=$.
Given that the generalization of $4\textnormal{-FEAS}_{\geq 0}$ to $\R$ is know to be complete for $\mathrm{NP}_\R$ which in turn is closely related to existential second-order logic, this does not come as a surprise.
We will take a closer look at the connection between $\mathrm{NP}_K$ and $\ESOK$ over positive semirings in Section~\ref{Fagin}.

\begin{example}[Feasibility of degree $4$ polynomials]
    \label{ex:4FEAS}
    The feasibility problem of real degree $4$ polynomials $4$-FEAS is known to be $\mathrm{NP}_\R$-complete~\cite{DBLP:conf/focs/BlumSS88}.
    In order to match our setting, we define the same problem for the positive reals, i.\,e.,  
    \[
        4\textnormal{-FEAS}_{\geq 0}\coloneqq \{f \in \R_{\geq 0}[x_1, \dots, x_n] \mid \deg(f) \leq 4, \exists \overline{y} \in \R_{\geq 0}^n f(\overline{y}) = 0\},
    \]
    where $\deg(f)$ denotes the degree of $f$.
    We define an $\mathrm{ESO}^=_{\R_{\geq 0}}$ sentence $\varphi_{4\textnormal{-FEAS}}$ that defines $4$-FEAS$_{\geq 0}$.
    For any $n \in \N$, let $[n]$ denote the $n$ long prefix of the positive natural numbers, i.\,e., $[n] \coloneqq \{1, \dots, n\}$.
    A degree $4$ polynomial $f \in \R_{\geq 0}[x_1, \dots, x_n]$ can be written as
    \begin{alignat*}{2}
        f(x_1, \dots, x_n) =  c_0 + &\smashoperator[l]{\sum_{(i_1) \in [n]}} &&c_{(i_1)} \cdot x_{i_1} +\phantom{ }\\
         & \smashoperator[l]{\sum_{(i_1, i_2) \in [n]^2}} &&c_{(i_1, i_2)} \cdot x_{i_1} \cdot x_{i_2} +\phantom{ }\\
        & \smashoperator[l]{\sum_{(i_1, i_2, i_3) \in [n]^3}} &&c_{(i_1, i_2, i_3)} \cdot x_{i_1} \cdot x_{i_2} \cdot x_{i_3} +\phantom{ }\\
        &\smashoperator[l]{\sum_{(i_1, i_2, i_3, i_4) \in [n]^4}} &&c_{(i_1, i_2, i_3, i_4)} \cdot x_{i_1} \cdot x_{i_2} \cdot x_{i_3} \cdot x_{i_4}
    \end{alignat*}
    where $c_0 \in \R_{\geq 0}$ and $c_{(i_1, \dots, i_j)} \in \R_{\geq 0}$ for all $(i_1, \dots, i_j) \in [n]^j$.

    Let $A = [n]$ and let $\tau = \{R_0^0, R_1^1, R_2^2, R_3^3, R_4^4\}$, where $R^x$ means that $\ar(R) = x$.
    We can represent $f$ as a $K$-interpretation $\pi \colon \lit_{A, \tau} \to K$, where $\evaluate{R_0()}{\pi} = c_0$ and $\evaluate{R_i(x_1, \dots, x_i)}{\pi} = c_{(x_1, \dots, x_i)}$.
    With this sort of representation, the following sentence defines $4$-FEAS$_{\geq 0}$:
    \begin{align*}
        \varphi_{4\textnormal{-FEAS}} \coloneqq & \exists Z^1 0 = \bigg(R_0() \lor \exists y_1 \Big(R_1(y_1) \land Z(y_1) \Big) \lor \phantom{ } \\
        & \exists y_1 \exists y_2 \Big(R_2(y_1, y_2) \land Z(y_1) \land Z(y_2)\Big) \lor \phantom{ } \\
        & \exists y_1 \exists y_2 \exists y_3 \Big(R_3(y_1, y_2, y_3) \land Z(y_1) \land Z(y_2) \land Z(y_3) \Big) \lor \phantom{ } \\
        & \exists y_1 \exists y_2 \exists y_3 \exists y_4 \Big(R_4(y_1, y_2, y_3, y_4) \land Z(y_1) \land Z(y_2) \land Z(y_3) \land Z(y_4)\Big)\bigg) 
    \end{align*}
    Since $R_0()$ yields $c_0$ and $R_i(y_1, \dots, y_i)$ yields $c_{(y_1, \dots, y_i)}$, the $K$-relation $Z$ determines a valuation of the variables $x_1, \dots, x_n$ in $f$.
    Thus, $\varphi_{4\textnormal{-FEAS}}$ is true, if and only if there exists an assignment to the variables of $f$ such that $f$ becomes zero.
    \begin{remark}
        Note that the quantification of the $0$-ary $K$-relation $R_0$ essentially corresponds to quantifying over the elements of the semiring.
    \end{remark}
\end{example}

\section{BSS Machines over Semirings and a Semiring Analogue to NP}\label{bss_section} 
BSS machines were introduced as a model for real computation by Blum, Shub and Smale~\cite{DBLP:conf/focs/BlumSS88} and were originally defined over rings.
The computation nodes of a BSS machine are usually formulated in terms of polynomial functions with coefficients from the underlying ring.
In case this ring is actually a field, also rational functions, i.\,e., divisions of polynomials, are permitted.
We modify this model to work over semirings $K$ and call it a $\BSSK$ machine. 

In BSS machines, the analogue of a Turing machine's tape is a state space consisting of infinite tuples with a dedicated centre.   
\begin{definition}
    For a set $A$, we let $A^* = \bigcup_{n\geq 0} A^n$.
    By $A_*$ we denote the set $A^{\Z}$ of all $x= (\dots,x_{-2}, x_{-1}, x_0 \textbf{.} x_1, x_2, \dots)$, where ``\textbf{.}'' marks the first position.
    We define two shift operations on $A_*$: shift left $\sigma_{\ell}$, where $\sigma_{\ell}(x)_i = x_{i+1}$, and its inverse, shift right $\sigma_r$, where $\sigma_r(x)_i = x_{i-1}$, for all $i \in \Z$.
\end{definition}
Based on this, we define our BSS machine model of computation over semirings. 
\begin{definition}[$\BSSK$ machines] \label{def:BSSK} 
    A $\BSSK$ machine $M$ consists of an input space $\mathcal{I}=K^*$, a state space $\mathcal{S}=K_*$ and an output space $\mathcal{O}=K^*$, together with a
    connected directed graph whose nodes are labelled by $1, \ldots, N$. 
    The nodes are of five different types as follows.
    \begin{description}
        \item[\normalfont\bfseries\sffamily Input node:] The node labelled by $1$ is the only input node. The node is associated with a next node $\beta(1)$ and the input mapping $g_I\colon \mathcal{I} \to \mathcal{S}$.
        \item[\normalfont\bfseries\sffamily Output node:] The node labelled by $N$ is the only output node. This node is not associated with any next node. 
        Once this node is reached, the computation halts, and the result of the computation is placed on the output space by the output mapping $g_O\colon \mathcal{S}\to \mathcal{O}$.
        \item[\normalfont\bfseries\sffamily Computation nodes:] A computation node $m$ is associated with a next node $\beta(m)$ and a mapping $g_m\colon \mathcal{S}\to \mathcal{S}$ such that for some $c \in K$ and $i,j,k\in \mathbb{Z}$ the mapping $g_m$ is identity on all coordinates different from $i$ and on current coordinate $i$ one of the following holds:
        \begin{itemize}
            \item $g_m(x)_i =x_j+x_k$ (addition),
            \item $g_m(x)_i = x_j\cdot x_k$ (multiplication),
            \item $g_m(x)_i = c$ (constant assignment).
        \end{itemize} 
        \item[\normalfont\bfseries\sffamily Branch nodes:] A branch node $m$ is associated with nodes $\beta^-(m)$ and $\beta^+(m)$. 
        Given $x\in \mathcal{S}$ the next node is $\beta^-(m)$ if $x_1 = x_2$ and $\beta^+(m)$, otherwise. 
        If $K$ is ordered, instead the next node is $\beta^-(m)$ if $x_1 \leq x_2$ and $\beta^+(m)$, otherwise.
        \item[\normalfont\bfseries\sffamily Shift nodes:] A shift node $m$ is associated either with shift left $\sigma_{\ell}$ or shift right $\sigma_r$, and a next node $\beta(m)$.
    \end{description}
    The \emph{input mapping} $g_I\colon \mathcal{I} \to \mathcal{S}$ places an
    input $(x_1, \ldots ,x_n)$ in the state
    \[
        (\ldots , 0, \underbrace{1, \dots, 1}_{n}, 0\textbf{.} x_1, \ldots ,x_n, 0,\ldots )\in \mathcal{S},
    \]
    where the size of the input $n$ is encoded in unary in the $n$ first negative coordinates.
    The \emph{output mapping} $g_O\colon \mathcal{S}\to \mathcal{O}$ maps a state to the string consisting of its first $\ell$ positive coordinates, where $\ell$ is the number of consecutive ones stored in the negative coordinates starting from the first negative coordinate.
    For instance,
    \[
        g_O((\ldots ,2,1,1,1,0\textbf{.}x_1, x_2,x_3,x_4,\ldots )) = (x_1, x_2,x_3).
    \]
    A \emph{configuration} at any moment of computation consists of a node $m\in \{1, \ldots ,N\}$ and a current state $x\in\mathcal{S}$.
    The (sometimes partial) \emph{input-output function} $f_M \colon K^*\to K^*$ of a machine $M$ is now defined in the obvious manner.
    A function $f \colon K^*\to K^*$ is \emph{computable} if $f=f_M$ for some machine $M$. A language $L \subseteq K^*$ is \emph{decided} by a $\BSSK$ machine $M$ if \[
        \bar{x} \in L \quad\text{if and only if} \quad f_M(\bar{x}) \neq 0,
    \]
    holds for all $\bar{x} \in L$.
\end{definition}

Analogously to classical complexity theory, we can argue about the complexity of a language over a semiring via the computation time of a $\BSSK$ machine deciding it.  
\begin{definition}
    Let $M$ be a $\BSSK$ machine.
    On any input $\bar{x} \in K^*$, the \emph{computation time} of $M$ on input $\bar{x}$ is the number of steps that $M$ takes before reaching the output node.
\end{definition}
Non-deterministic decisions are realised using the same idea of verification as in the classical model of Turing machines. 
In our case, additional $K$-values are guessed and added after the input to the $\BSSK$ machine by modifying the input mapping.

\begin{definition}\label{def:nondeterministic_BSS_machine}
    A language $L \subseteq K^*$ is \emph{decided non-deterministically} by a $\BSSK$ machine $M$, if for all $\bar{x} \in K^*$,
    \[
        \bar{x} \in L  \text{ iff there exists an $\bar{x}' \in K^*$ (called the guess) such that } f_M(\bar{x}, \bar{x}') \neq 0,
    \]
    where the modified input mapping $g_I' \colon \mathcal{I} \to \mathcal{S}$, which places $((x_1, \dots, x_n), (x_1', \dots, x_m'))$ in the state
    \[
        (\dots, 0, \underbrace{1, \dots, 1}_{m}, 0, \underbrace{1, \dots, 1}_{n}, 0. x_1, \dots, x_n, x_1', \dots, x_m', 0, \dots) \in \mathcal{S}
    \]
    is used.
\end{definition}
In accordance with the work of Blum, Shub and Smale \cite{DBLP:conf/focs/BlumSS88} we call such a machine a \emph{non-deterministic $\BSSK$ machine}.
Based on this definition, we are able to define a semiring analogue to the class $\NP$ via the model of non-deterministic $\BSSK$ machines.

\begin{definition}
    The class $\NPK$ consists of all languages $L \subseteq K^*$ for which there exists a $\BSSK$ machine $M$ that runs in polynomial time such that $M$ decides $L$ non-deterministically.
\end{definition}
Similarly, we define the notion of polynomial time many-to-one reductions with respect to functions computable in polynomial time by $\BSSK$ machines.
\begin{definition}[$\leqpmk$]\label{def:leqpmk}
    Let $P, Q \subseteq K^*$.
    We then write $P \leqpmk Q$ if there exists a function $f \colon K^* \to K^*$ such that $f$ is computable in polynomial time by a $\BSSK$ machine and
    for all $\bar{x} \in K^*$ it holds that $x \in P$ iff  $f(\bar{x}) \in Q$.
\end{definition}
In line with the previous definitions, we now define the analogue of $\NP$-completeness over a semiring $K$.
\begin{definition} \label{def:npk_complete}
    A language $L \subseteq K^*$ is \emph{$\NPK$-complete} if $L \in \NPK$ (membership in $\NPK$) and we have that $L' \leqpmk L$, for all languages $L' \in \NPK$ (hardness for $\NPK$).
\end{definition}

\section{A Version of Fagin's Theorem}\label{Fagin}

We now present a version of Fagin's theorem in the context of semiring semantics, which is used to characterise our class $\NPK$ via existential second-order logic over a semiring. 
The proof follows the classical structure, but has been adapted to work with our model of $\BSSK$ machines. 
From the logical perspective, we consider ordered $K$-interpretations in which a single unary $K$-relation $I$ is used to encode an input in $K^*$. 
Before relating the complexity class $\NPK$ and the logic $\ESOK$, it is necessary to consider how formulae and $K$-interpretations are encoded.

\begin{remark}\label{rem:encoding}
    A $K$-interpretation can be encoded as elements of $K^*$ (see~\cite[Sect.~2.4]{Jelia25}).
    Similarly, we can encode formulae of $\FOK$ and $\ESOK$ as elements of $K^*$.
    For these encodings, we follow the convention that elements of the semiring $K$ are encoded
    as they are, i.\,e., each $x \in K$ is encoded by $x$ itself, and other information
    is encoded using binary encodings with the semiring elements $0$ and $1$.
    Then, for instance, symbols for connectives, quantifiers, variables and arithmetic operations are
    encoded as binary strings, whereas constants corresponding to elements of $K$ are encoded as
    their interpretations in $K$.
\end{remark}
For brevity, from now on, we overload the notation $\ESOK$ and use it to denote the set of formulae, as well as the sets of $K$-interpretations $\{\,\pi \mid \evaluate{\varphi}{\pi} \neq 0\,\}$ and encodings of $K$-interpretations 
$\{\,\langle \pi \rangle \mid \evaluate{\varphi}{\pi} \neq 0\,\}$
\emph{defined} by $\ESOK$ sentences $\varphi$.

\begin{lemma}\label{lem:eso_in_np}
    $\ESOK \subseteq \NPK$
\end{lemma}
\begin{proof}
    Let $\phi \coloneq \exists R_1 \dots \exists R_k \psi $ be an existential second-order formula with quantified relation symbols $R_1, \dots, R_k$. 
    Per definition, we have that $\evaluate{\phi}{\pi} \neq 0$ for a $K$-interpretation $\pi$ if there exists a $K$-interpretation $\pi'$ for the symbols $R_1, \dots, R_k$ such that $\evaluate{\psi}{\pi \cup \pi'} \neq 0$.
    There exists a $\BSSK$ machine $M$ that deterministically evaluates $\psi$ in polynomial time given the encoding of a $K$-interpretation $\langle \pi\cup\pi' \rangle$ as input (see~\cite[Corollary 27]{Jelia25} ).
    Thus, there is a non-deterministic $\BSSK$ machine $M'$ that, with input $\langle \pi \rangle$, guesses $\langle \pi' \rangle$, and accepts if and only if $\evaluate{\phi}{\pi} \neq 0$.    
    Furthermore, it holds that $\evaluate{\Phi}{\pi} \neq 0$ if only if $\evaluate{\Psi}{\pi \cup \pi'} \neq 0$ if and only if $f_{M'}( \langle \pi \rangle, \langle \pi' \rangle) \neq 0$.
\end{proof}

Given the previous lemma, we can now provide an exact characterisation of $\NPK$ via existential second-order formulae.
\begin{theorem} \label{thm:fagin}
   $\ESOK = \NPK$
\end{theorem}
\begin{proof}
    Since Lemma~\ref{lem:eso_in_np} gives us $\ESOK \subseteq \NPK$, the converse containment, $\NPK \subseteq \ESOK$, is left to prove. 
    As mentioned in the beginning of this section, machine inputs of length $n$ are interpreted as ordered $K$-interpretations of size $n$ with one unary $K$-relation $I$, such that $I(x)$ gives the $x$-th element of the machine input.

    Let $M$ be the $\BSSK$ machine that runs in time $T(n) \leq n^z - 1$ for some $z\in \N$ on input $I$ of length $n$.
    We use $\bar{t}=(t_1, \dots, t_z) \in A^z$ and $\bar{p} = (d, p_1, \dots, p_z) \in A^{z+1}$, where $A = \{a_0, \dots, a_{n - 1}\}$ is our domain, to identify the time step and position of $M$. 
    As $\BSSK$ machines have a tape that is infinite in both directions we use $q \in \{a_0, a_1\} \subseteq A$ to distinguish between negative and positive positions.
    We say that $\bar{p}_\ell$, where $\ell \in \Z$, denotes the $\ell$th position.
    For example, consider $A = \{a_0, a_1\}$ and $z = 2$, then we have that $\bar{p}_1 = (a_1,a_0,a_1)$ and $\bar{p}_{-3} = (a_0,a_1,a_1)$, where the latter is the smallest expressible position $\bar{p}_{\min}$ in this example.
    
    We define an $\ESOK$ formula $\Phi$ which describes that there exists a computation such that $M$ accepts. 
    We will use the following relations to describe a $\BSSK$ machine $M$:
    \begin{itemize}
        \item $V(\bar{t}, \bar{p})$ describes the tape content at time step $\bar{t}$ and position $\bar{p}$
        \item $Q_s(\bar{t}) = \begin{cases}
            1, & \text{if $M$ is at node $s$ at time step $\bar{t}$}\\
            0, & \text{else}
        \end{cases}$
    \end{itemize}
    We define the second-order formula  
    \[\Phi \dfn \exists V^{2z+1}, Q^z_{s_1}, \dots, Q^z_{s_{N}} \varphi,\]
    that existentially quantifies the relation $V(\overline{t}, \overline{p})$ with arity $2z+1$ and $Q_{s}(\overline{t})$ with arity $z$ for each node $s \in S$. 
    The $\FOK$ formula $\varphi$ describes our machine $M$.
    Note that, except for a few cases, all connectives can be replaced by Boolean connectives, as defined in Definition~\ref{def:bool}. 
    This is because we only use them in the context of binary comparisons.

    Let $S$ be the set of all nodes of $M$. 
    As defined in Definition \ref{def:BSSK}, a node $s$ is of one of five different types. 
    We divide $S$ accordingly as described in the following.
    Let $S_\text{add}$/$S_\text{mult}$/$S_{\text{const}}$ be the sets of all addition/multiplication/constant nodes.
    These are associated with values $i_s,j_s, k_s \in \Z$ and a $c_s\in K$ via the mapping $g$ for all nodes $s$ in $S_\text{add}$/$S_\text{mult}$/$S_{\text{const}}$.
    They describe that the computation is applied to the operands at positions $\bar{p}_{j_s}$ and $\bar{p}_{k_s}$.
    The result is stored at position $\bar{p}_{i_s}$. 
    For constant nodes, $c_s$ is the value that is assigned at position $\bar{p}_{i_s}$.
    The set $S_{\text{branch}}$ includes all branching nodes. 
    Those are associated with the two positions they may branch to, the current position $\bar{p}_1$, the next position $\bar{p}_2$ and the next node functions $\beta^-$ and $\beta^+$ of $M$ that give the node to branch to.
    The set $S_{\text{shift}}$ includes the shift nodes which perform either a right or left shift of the tape.
    The input node is labelled with 1 and the output node with $N$.

    For readability, we use restricted quantifiers $(\forall x.\alpha)\varphi \dfn (\forall x)(\alpha \bto \varphi)$ and $(\exists x.\alpha)\varphi \dfn (\exists x)(\alpha \band \varphi)$.
    Now, define $\varphi$ as follows:
    \[\varphi \dfn \forall \bar{t} \forall \bar{p} \exists \bar{t}'.(\bar{t}'=\bar{t}+1) (\varphi_{\text{input}} \land \varphi_{\text{output}} \land \varphi_{\text{arithmetic}} \land \varphi_{\text{branching}} \land \varphi_{\text{constraint}}).\]
    We write $\bar{t}'=\bar{t}+1$ to define the next time step which is easily expressible using the order relation. 
    Likewise, we will write $\bar{p}'=\bar{p}+1$ and $\bar{p}'=\bar{p}-1$, to denote a shift to the position by one in either direction. 
    For simplicity, we assume that $\bar{p}_{\max}+1=\bar{p}_{\max}$ and $\bar{p}_{\min}-1=\bar{p}_{\min}$.
    Although information is technically lost at the edge of the bounded tape, the tape is large enough that this cannot affect the computation due to time constraints of the machine.

    The subformula $\varphi_{\text{input}}$ encodes the input configuration.
    The formulae $\varphi_{\text{arithmetic}}$ and $\varphi_{\text{branching}}$ describe the current configuration of $M$ after a move. 
    The next configuration depends on the current node $s$, in particular, the kind of operation performed at this node. 
    The constraints of our machine model are represented in $\varphi_{\text{constraint}}$, and $\varphi_{\text{output}}$ describes the termination of the computation.
    
    To correctly read our input into the input state, we define:
\begin{alignat}{1}
    \varphi_{\text{input}} \coloneq &(Q_1(\bar{t}) =1) \bto \bar{t} = \bar{t}_0 \,\land \nonumber\\
    &\forall x \left(V(\bar{t}, (a_1, a_0, \dots, a_0, x)) = I(x) \land V(\bar{t}, (a_0, a_0, \dots, a_0, x)) = 1 \right) \land\\
    &V(\bar{t}, (a_0, a_0, \dots, a_0, a_1, a_0)) = 0 \land V(\bar{t}, (a_1, a_0, \dots, a_0, a_1, a_0)) = 0 \, \land\\
    &\exists \bar{d} \bigl( \forall \bar{d}'. ((a_0, \dots a_0,a_1, a_0) < \bar{d}' \land \bar{d}' \leq \bar{d}) \, (V(\bar{t}, (a_0, \bar{d}'))=1) \, \land\\
    &\phantom{\exists \bar{d} \bigl(} \forall \bar{d}'.(\bar{d}'>\bar{d})\, (V(\bar{t}, (a_0, \bar{d}')) = 0 \land V(\bar{t}, (a_1, \bar{d}')) = 0)\bigr),
\end{alignat}
where $a_0, a_1 \in A$ denote negative and positive positions respectively. 
We existentially quantify over the combined length $d=n+m$ of input and guess.
Our input $I$ with $|I|=n$ is read and written on the first positions on the tape.
As stated in Definition \ref{def:nondeterministic_BSS_machine}, the first $n$ negative positions contain ones to encode the length $n$ of the input (1).
A zero separates the input from the guess and the ones encoding the length of the input from the encoding of the length of the guess (2).
A one is placed on the negative positions between $n+1$ and $d$ to encode the length of the guess (3).
To ensure that the guess corresponds to this length all negative/positive positions on the tape smaller/greater than $d+1$ are zero (4). 
    
    For arithmetic kinds of nodes, i.\,e., addition, multiplication and constant assignment we define
    $\varphi_{\text{arithmetic}}$ as:
    \begin{alignat*}{2}
        \varphi_{\text{arithmetic}} &\coloneq \notag\\
        & \smashoperator[l]{\bigwedge_{s \in S_\text{mult}}} &&\bigl((Q_s(\bar{t})=1) \bto V(\bar{t}', \bar{p}_{i_s})= V(\bar{t}, \bar{p}_{j_s}) \land V(\bar{t}, \bar{p}_{k_s})\bigr) \land\\
         & \smashoperator[l]{\bigwedge_{s \in S_\text{add}}} &&\bigl((Q_s(\bar{t})=1) \bto V(\bar{t}', \bar{p}_{i_s})= V(\bar{t}, \bar{p}_{j_s}) \lor V(\bar{t}, \bar{p}_{k_s})\bigr) \land\\
        & \smashoperator[l]{\bigwedge_{s \in S_{\text{const}}}} &&\bigl((Q_s(\bar{t})=1) \bto V(\bar{t}', \bar{p}_{i_s}) = c_s\bigr),
    \end{alignat*}
    where $i_s,j_s,k_s \in \Z$ to denote the $i_s$th, $j_s$th and $k_s$th position associated with the arithmetic node $s$ as specified in Definition \ref{def:BSSK}.    
    Notice that in this formula we need $\lor$ and $\land$ as defined for the semiring to add and multiply $V(\bar{t}, \bar{p}_{j_s})$ and $V(\bar{t}, \bar{p}_{k_s})$.
    
    For branching, we distinguish between the branching, non-branching, and shift nodes and define 
    \begin{alignat*}{2}
        \varphi_{\text{branching}} &\coloneq \notag\\
        & \smashoperator[l]{\bigwedge_{s \in S_{\text{shift}}}} &&\bigl((Q_s(\bar{t})=1) \bto \exists \bar{p}' . \bigl(\bar{p}' = \bar{p} \star_s 1) (V(\bar{t}', \bar{p}) = V(\bar{t}, \bar{p}')\bigr)\bigr) \land \\
        & \smashoperator[l]{\bigwedge_{s \in S\setminus S_{\text{branch}}}} &&\bigl((Q_s(\bar{t})=1) \bto Q_{\beta(s)}(\bar{t}') = 1\bigr) \land\\
        & \smashoperator[l]{\bigwedge_{s \in S_{\text{branch}}}} &&\bigl((Q_s(\bar{t})=1) \bto (V(\bar{t}, \bar{p}_1) \leq V(\bar{t}, \bar{p}_2) \bto Q_{\beta^+(s)}(\bar{t}') = 1) \\
        & && \land (V(\bar{t}, \bar{p}_1) \nleq V(\bar{t}, \bar{p}_2) \bto Q_{\beta^-(s)}(\bar{t}') = 1)\bigr), 
    \end{alignat*}
    where $\star_s \in \{-, +\}$ to denote either a right or left shift depending on $s$. 
    
    The constraints of our machine model are enforced with the following formula:
    \begin{alignat}{2}
        \varphi_{\text{constraint}} &\coloneq \notag\\
       &\smashoperator[l]{\bigwedge_{s \in S_{\text{branch}}}} &&\bigl((Q_s(\bar{t})=1) \bto (V(\bar{t}', \bar{p})= V(\bar{t}, \bar{p}))\bigr) \land\\
       &\smashoperator[l]{\bigwedge_{s \in (S_\text{add} \cup S_\text{mult} \cup S_{\text{const}})}} &&\bigl(\bar{p} \neq \bar{p}_i \bto(V(\bar{t}, \bar{p}) = V(\bar{t}', \bar{p}))\bigr) \land\\
       &\smashoperator[l]{\bigwedge_{s \in S}} && \bigl((Q_s(\bar{t})=1) \bto \bigwedge_{s \neq s' \in S} Q_{s'}(\bar{t})=0\bigr).
    \end{alignat}
    This corresponds to the constraints:
    \begin{enumerate}[label=(\arabic*)]
        \item The content of the tape does not change for branching nodes.
        \item No more than one cell can change per step for arithmetic nodes.
        \item The machine remains in a single node at a given time point, with only one node change occurring per step.
    \end{enumerate}
    
We require that the node and the tape content remain unchanged after the final node is reached and the machine halts. 
The final node of $M$ is the output node labelled $N$.
We require the tape to be of the form $\dots z10\textbf{.} x_1\dots$, where $z \neq 1$ and $x_1 \neq 0$ in order to comply with our definition of the output readout of a $\BSSK$ machine. 
This yields the formula 
\begin{alignat*}{1}
    \varphi_{\text{output}} \dfn (Q_N(\bar{t})=1) \bto (&Q_N(\bar{t}')=1 \land V(\bar{t}', \bar{p})= V(\bar{t}, \bar{p}) \land V(\bar{t}, \bar{p}_{-2}) \neq 1 \land \\
    &V(\bar{t}, \bar{p}_{-1}) =1 \land V(\bar{t}, \bar{p}_1) \neq 0).
\end{alignat*}
The proof of correctness is then straightforward.
\end{proof}

\section{A Version of Cook's Theorem}\label{Cook}
Similar to the classical setting, we can now prove a version of Cook's theorem over semirings, as in Definition \ref{def:satk}.
That is, we show that the satisfiability problem of $\PLK$ is also complete for the class of non-deterministic polynomial time in the semiring setting.
Note that, as mentioned in \Cref{rem:encoding}, we use encodings for formulae and assignments instead of $K$-interpretations.
We will also overload the notation $\PLK$ to denote sets of such formulae, assignments $\{s \mid \evaluate{\varphi}{s} \neq 0\}$, and encodings of assignments $\{\langle s \rangle \mid \evaluate{\varphi}{s} \neq 0\}$.
\begin{lemma}\label{lem:satkinnpk}
    $
        \SATK \in \NPK.
    $
\end{lemma}
\begin{proof}
    Let $M$ be a non-deterministic $\BSSK$ machine with input $\langle \varphi \rangle$, the encoding of a $\PLK$ formula, that guesses $\langle s \rangle$, the encoding of an assignment for $\varphi$.
    This machine $M$ then computes $f_M(\langle \varphi \rangle, \langle s \rangle) = \evaluate{\varphi}{s}$ in polynomial time, bounded by the size of the encoding $|\langle \varphi \rangle|$ of the formula $\varphi$.

    Now, we have that $\varphi \in \SATK$ if and only if there exists an assignment $s$ such that $\evaluate{\varphi}{s} \neq 0$ if and only if there is a guess $\langle s \rangle$ such that $f_M(\langle \varphi \rangle, \langle s \rangle) \neq 0$.
\end{proof}
Given the containment of $\SATK$ in $\NPK$, it is left to show that every language $L \in \NPK$ is polynomial time many-to-one reducible to $\SATK$ as by Definition \ref{def:npk_complete}.
As the proof follows the same idea as the proof of Fagin's theorem for semirings (Theorem \ref{thm:fagin}), we only present the idea of the proof here. 
\begin{restatable}{theorem}{cook} \label{thm:cook}
    $\SATK$ is $\NPK$-complete.
\end{restatable}

\begin{proof}
As in the proof of Theorem \ref{thm:fagin}, we construct a formula $\varphi$ that describes the computation of a non-deterministic $\BSSK$ machine $M$ on some arbitrary input $x$, such that $\varphi$ is satisfiable if and only if $M$ accepts $x$. 
To do so, we define propositions (instead of relations) that describe the content of the tape and the current computation node at a given time:
\begin{itemize}
        \item $v_{t,p}$ describes the content of the tape at time $t$ and position $p$
        \item $q_{t,s} = \begin{cases}
            1, & \text{if $M$ is at node $s$ at time $t$}\\
            0, & \text{else}
        \end{cases}$
\end{itemize}
From there on, we follow the proof of Theorem \ref{thm:fagin}, constructing formulae to verify that the machine is correctly initialised and that the values at computation nodes change according to the rules of the machine.
For instance, the following formula describes that the machine is initially at the first computation node and nowhere else:
\[ \varphi^0_q \coloneqq (q_{0, 1} = 1) \land \bigwedge_{s \in S\setminus\{ 1\}} q_{0,s} = 0 \]
We concatenate this formula together with similar checks to ensure a correct computation and termination, which gives us the desired $\PLK$-formula $\varphi$. 
The full proof can be found in Appendix \ref{app:cook}.
\end{proof}

\section{Non-arithmetic Reductions and the Existential Theory of K}\label{ETKand}

In this section, we define novel non-arithmetic reductions and use them to show that the existential
first-order theory of a semiring $K$ is complete for the Boolean part of the complexity class $\NPK$.
We also generalize the aforementioned result to $\BSSK$ machines that operate on subsets of $K$.
For this, we turn our attention to $\BSSK$ machines $M$ that correspond to
polynomial-time computable input-output functions
of the forms $f_M \colon X^* \to K^*$ and $f_M \colon X^* \to X^*$ for some subsets $X \subseteq K$.
However, without loss of generality,
we may assume that these machines run in polynomial time for all inputs in $K^*$.
This simplifying observation is an application of using a binary counter to
restrict the number of computation steps before always halting. 
Missing proofs of this section can be found in \Cref{app:etk}.

\begin{restatable}{lemma}{timer} \label{lem:timer}
Let $X \subseteq K$ and let $f \colon X^* \to K^*$ be computed by a $\BSSK$ machine $M$
that runs in polynomial time on inputs in the set $X^*$.
Then there is a $\BSSK$ machine $M'$ such that
$f_{M'} \upharpoonright X^* = f$ and that $M'$ runs in polynomial time on all inputs in $K^*$.
Furthermore, the machine $M'$ is limited to the machine constants of $M$ together with $0$ and $1$.
\end{restatable}

\begin{definition}\label{def:npkx}
Let $\{0, 1\} \subseteq X \subseteq K$.
The class $\NPKXX$ consists of all languages $L \subseteq X^*$ for which
there exists a $\BSSK$ machine $M$ that runs in polynomial time
such that $M$ decides $L$ non-deterministically among inputs in $X^*$
and uses only machine constants in the set $X$.
That is, for every input $\bar{x} \in X^*$ it holds that $\bar{x} \in L$, if and only if for some
guess $\bar{x}' \in K^*$ it is true that $f_M(\bar{x}, \bar{x}') \neq 0$.
\end{definition}

\begin{definition}
For $L \subseteq K^*$ and $L' \subseteq X^*$, we write
$L' \leqpmkxx L$ if there is a function $f \colon X^* \to K^*$ that is computable in
polynomial time by a $\BSSK$ machine using only machine constants in the set $X$
such that for all $\bar{x} \in X^*$ it holds that $\bar{x} \in L'$ if and only if $f(\bar{x}) \in L$.
\end{definition}

In particular, $\NPKX{K} = \NPK$, and for every $L, L' \subseteq K^*$ the conditions
$L' \leqpmk L$ and $L' \leqpmkx{K} L$ are equivalent.
We do not assume that the subsets $X \subseteq K$ we consider
are closed under the arithmetic operations of the semiring $K$, but instead we focus
on reductions based on $\BSSK$ machines that use arithmetic operations only for copying
arbitrary elements of the semiring without changing them, and for assigning constant values.

\begin{definition}\label{def:nonarithmetic}
A $\BSSK$ machine $M$ is said to be \emph{non-arithmetic}, if for every input $\bar{x} \in K^*$
it holds that whenever an arithmetic computation node is reached during the computation,
this computation node corresponds to either
a constant assignment or an arithmetic operation $a \star b$ for $\star \in \{+, \cdot\}$
in such a manner that at least one of the operands $a$ and $b$
is the neutral element of the operation $\star$ in the semiring $K$.
Furthermore, a computable function $f$ is \emph{non-arithmetic},
if it is computed by some non-arithmetic machine.
For $X \subseteq K$ and $P, Q \subseteq X^*$,
a reduction $Q \leqpmkxx P$ is called \emph{non-arithmetic}, if the computable function
corresponding to the reduction is non-arithmetic.
\end{definition}

\begin{lemma}\label{lem:stillinxstar}
Let $\{0, 1\} \subseteq X \subseteq K$ and let $M$ be a non-arithmetic $\BSSK$ machine that only uses machine
constants in the set $X$. Then, for every input $\bar{x} \in X^*$, during the computation only
those elements of $K$ that are in the subset $X$ may appear, and it holds that $f_M(\bar{x}) \in X^*$.
\end{lemma}
\begin{proof}
    Since a $\BSSK$ machine can change the semiring elements in the state space of the machine only while executing an arithmetic computation node, using induction on the computation steps of a non-arithmetic $\BSSK$ machine yields the stated claim. 
    Note that $0$ is also used to fill the state space, and $1$ is used to encode the length of the input.
\end{proof}

In the case that $X$ is finite, the previous lemma leads to an alternative
approach of using ordinary Turing machines instead of $\BSSK$ machines.

\begin{restatable}{lemma}{nonarithmeticturing} \label{lem:nonarithmeticturing}
Let $X \subseteq K$ be a finite subset so that $\{0, 1\} \subseteq X$.
Let $P, Q \subseteq X^*$ and let $Q \leqpmkxx P$ be a non-arithmetic reduction.
Then we have $Q \leqpm P$, that is, there is a Turing machine with input alphabet $X$ computing a reduction from $Q$ to $P$ in polynomial time.
\end{restatable}

We give two definitions for completeness based on non-arithmetic reductions.
The first one is a non-arithmetic version $\NPKXX$-completeness,
and the other one is an even stronger form of non-arithmetic $\NPK$-completeness.

\begin{definition}
Let $\{0, 1\} \subseteq X \subseteq K$ and $L \subseteq X^*$. We say that $L$ is
\emph{$\NPKXX$-complete under non-arithmetic reductions}, if $L \in \NPKXX$
and if for every language
$L' \in \NPKXX$ there is a non-arithmetic reduction $L' \leqpmkxx L$.
\end{definition}

\begin{definition}
A language $L \subseteq K^*$ is said to be \emph{$\NPK$-complete under non-arithmetic reductions respecting machine constants}, if the membership $L \in \NPK$ is witnessed by
polynomial-time $\BSSK$ machine that only uses machine constants $0$ and $1$,
and if for every $L' \in \NPK$
and for every polynomial-time $\BSSK$ machine $M$ the following holds:
if $M$ decides $L'$ non-deterministically, then there is a polynomial-time non-arithmetic
$\BSSK$ machine $M'$ that is limited to only use
the machine constants of $M$ together with $0$ and $1$
in such a manner that for every $\bar{x} \in K^*$
it is true that $\bar{x} \in L'$ if and only if $f_{M'}(\bar{x}) \in L$.
\end{definition}

\begin{lemma}\label{lem:npkxcompletesub}
Let $\{0, 1\} \subseteq X \subseteq K$ and let $P \subseteq K^*$ be $\NPK$-complete
under non-arithmetic reductions that respect machine constants.
Then the problem $P \cap X^*$ is $\NPKXX$-complete under non-arithmetic reductions.
\end{lemma}
\begin{proof}
By \Cref{def:npkx} and the assumed witness for the membership $P \in \NPK$,
it holds that $P \cap X^* \in \NPKXX$.
For the claimed $\NPKXX$-hardness under non-arithmetic reductions, let $Q \in \NPKXX$ be arbitrary.
Let $M$ be some polynomial-time $\BSSK$ machine that non-deterministically decides $Q$ among
inputs in $X^*$ and that only uses machine constants in the set $X$.
Let $Q' \subseteq K^*$ be the language that $M$ decides among inputs in $K^*$.
In particular, $Q' \cap X^* = Q$.
By assumption,
there is a non-arithmetic $\BSSK$ machine $M'$ that runs in polynomial time and uses only
machine constants in the set $X$ so that for every $\bar{x} \in K^*$, we have
$\bar{x} \in Q'$ if and only if $f_{M'}(\bar{x}) \in P$.
Then, it follows from \Cref{lem:stillinxstar}, that for every $\bar{x} \in X^*$
it holds that $\bar{x} \in Q$ if and only if $f_{M'(\bar{x})} \in P \cap X^*$.
Therefore, $P$ is $\NPKXX$-complete under non-arithmetic reductions.
\end{proof}

For the next result, let $\SATKflat$
be the restriction of $\SATK$ to those propositional formulae
that do not have nested formula (in)equalities.
For this definition, we consider the shorthand notations for the negated formula inequalities
$\neq$ and $\nleq$ to be single logical symbols.

\begin{theorem}\label{thm:satkrespect}
The problems $\SATK$ and $\SATKflat$ are $\NPK$-complete under non-arithmetic reductions
that respect machine constants.
\end{theorem}
\begin{proof}
Although not explicitly stated,
the claims in the statement have been demonstrated in the proofs of \Cref{lem:satkinnpk,thm:cook}.
In the proof of \Cref{lem:satkinnpk}, machine constants for elements other than $0$ and $1$
are not used when a guess for a satisfying assignment is verified.
The proof of \Cref{thm:cook} shows that, in addition to $\SATK$ being $\NPK$-complete,
also $\SATKflat$ is $\NPK$-complete. Even though we have used Boolean connectives in
the sense of \Cref{def:bool}, the use of nested formula (in)equalities can be eliminated.
Specifically, each subformula of the form $\phi \bto \psi$ can be replaced with
$\phi \neq 0 \lor \psi$ without changing the satisfiability of the resulting formula.
The corresponding reductions can be accomplished using non-arithmetic machines
that satisfy the claimed condition on respecting machine constants.
\end{proof}

Next we consider the semiring $K$ as a structure
in the context of ordinary first-order logic.
   \emph{The existential theory of $K$} is defined as the set of existential sentences of first-order logic (in the vocabulary $\{+,\cdot,0,1\}$) that are true in $K$, i.\,e., as the set $\ETK=\{\,\phi \in \FO \mid K\models \phi\,\}$, where the standard Tarskian semantics of truth is applied.
   To be more precise, the syntax of this logic is defined otherwise as in \Cref{def:fok}, but we leave out universal quantification and formula (in)equalities, and only allow the constants $0$ and $1$.
   Then, $\EK$ is defined as the complexity class that consists of all those Boolean languages that have a polynomial-time many-one reduction to $\ETK$; in short,
   $\EK\coloneqq \{\,L\subseteq \{0,1\}^*\mid L \leqpm \ETK\,\}$.

Additionally,
For every subset $X \subseteq K$ so that $\{0, 1\} \subseteq X$,
we define $\ETKXX$ as the set of first-order sentences that are true in the structure
that is obtained from $K$ by interpreting each semiring element
in the set $X$ as an individual constant.
In particular, $\ETKX{\{0, 1\}} = \ETK$.

\begin{theorem}\label{thm:etkxcomplete}
$\ETKXK$ is $\NPK$-complete under non-arithmetic reductions that respect machine constants.
\end{theorem}
\begin{proof}
The membership of $\ETKXK$ in $\NPK$ can be demonstrated in a similar manner to the proof
of \Cref{lem:satkinnpk}, which is, by guessing a satisfying assignment for the existentially
quantified variables and then verifying it in polynomial time. Only the elements $0$ and $1$
are used as machine constants.

For the hardness claim, note that
since the composition of two non-arithmetic reductions is a non-arithmetic reduction,
by \Cref{thm:satkrespect}
it is enough to show that there is
a non-arithmetic reduction $\SATKflat \leqpmk \ETKXK$ that does not use
other machine constants besides $0$ and $1$.
This reduction can be realized as follows:
Let $\phi$ be the propositional formula that is given as input.
Let $g$ be an injective function that maps
each negated and non-negated propositional literal of $\phi$
to a first-order variable.
Within formula (in)equalities, the connectives $\land$ and $\lor$ are replaced
with the arithmetic operations $\cdot$ and $+$, respectively,
and each propositional literal $l$ is replaced with $g(l)$.
Outside the scope of formula (in)equalities, each propositional literal $l$
is replaced with $g(l) \neq 0$.
Constants are kept unchanged.
At this point, we have a first-order formula $\psi$
without quantifiers. Then, let $\theta \coloneq \exists x_1 \dots \exists x_k \psi$,
where $x_1, \dots x_k$ are the newly introduced first-order variables of $\psi$.
Using arguments based on induction, it can be shown that
the propositional formula $\phi$ is satisfiable if and only if $K \models \theta$.
The first-order sentence $\theta$ can be constructed using a non-arithmetic $\BSSK$ machine.
\end{proof}

Since $\ETKXX = \ETKXK \cap X^*$ whenever $\{0, 1\} \subseteq X \subseteq K$,
the following corollary follows from \Cref{lem:npkxcompletesub,thm:etkxcomplete}.

\begin{corollary}
Let $\{0, 1\} \subseteq X \subseteq K$.
Then the problem $\ETKXX$ is $\NPKXX$-complete under non-arithmetic reductions.
\end{corollary}

In the previous corollary, in the case that the subset $X$ is finite, we can use
\Cref{lem:nonarithmeticturing} to obtain a result about complexity classes
defined by using ordinary Turing machines.

\begin{corollary}
Let $X \subseteq K$ be a finite subset so that $\{0, 1\} \subseteq X$. Then $\ETKXX$
is complete for $\NPKXX$ under polynomial-time reductions using Turing machines
with input alphabet $X$.
In particular, $\EK = \NPKX{\{0, 1\}}$.
\end{corollary}

By using the notation of, e.\,g., \cite{10.1145/1007352.1007425},
our definition of $\NPKX{\{0, 1\}}$ coincides with the class that would be denoted by $\BP{\NPK^0}$,
the Boolean part of $\NPK^0$.
Here, $\NPK^0$ is the restricted version of the definition of $\NPK$ to $\BSSK$ machines
that are only allowed to use the semiring values $0$ and $1$ as machine constants.
Then, $\BP{\NPK^0}$ is defined as $\{\,L \cap \{0, 1\}^* \mid L \in \NPK^0\,\}$.

\section{Conclusion and Future Work}

We have given a logical characterization of the complexity class $\NPK$, i.\,e., non-deterministic polynomial time for $\BSSK$ machines, via a logic $\ESOK$, which is a version of existential second-order logic with semiring semantics. 
We have also shown that $\SATK$, the satisfiability problem of propositional logic with semiring semantics, is complete for $\NPK$.  Furthermore, we have shown that the true existential first-order theory of a semiring $K$ is a complete problem for the so-called Boolean part of $\NPK$, and an analogous result holds more generally for
$\BSSK$ machines that operate on subsets of $K$.

There are several promising opportunities for further research. In light of the counterparts of Fagin's and Cook's theorems presented in the semiring context, future research could focus on studying the P-versus-NP problem for specific semirings. 

The PCP theorem~\cite{DBLP:conf/stoc/BabaiFLS91,DBLP:journals/jacm/AroraLMSS98} establishes that every language in NP has probabilistically checkable proofs verifiable with logarithmic randomness and a constant number of queries, a cornerstone result linking proof verification to hardness of approximation. 
An intriguing open direction is whether an analogue could be formulated in more general algebraic frameworks, such as semirings. 
This raises the question of whether notions of proof checking, approximation resistance, or hardness results could meaningfully extend beyond fields into semiring-based complexity settings, and what structural barriers might prevent such a generalization.

It would also be interesting to study logics of this paper from the point of view of other logical formalisms that can accommodate semiring structures. 
One possible idea would be to relate the extensions of first-order logic considered in this work to metafinite structures and to study their finite model theoretic properties.
Another possible avenue is the contextualization of the logics defined in this paper with respect to logics over semirings, that do not have access to formula inequalities and/or constants, such as the first-order logic with semiring semantics introduced by Grädel and Tannen~\cite{abs-1712-01980, Grädel2025}.

\bibliography{ref_url,references}

\begin{thebibliography}{10}

\bibitem{10.1145/3643027}
Mahmoud Abo~Khamis, Hung~Q. Ngo, Reinhard Pichler, Dan Suciu, and Yisu~Remy Wang.
\newblock Convergence of datalog over (pre-) semirings.
\newblock {\em J. ACM}, 71(2), April 2024.
\newblock \href {https://doi.org/10.1145/3643027} {\path{doi:10.1145/3643027}}.

\bibitem{DBLP:journals/jacm/AroraLMSS98}
Sanjeev Arora, Carsten Lund, Rajeev Motwani, Madhu Sudan, and Mario Szegedy.
\newblock Proof verification and the hardness of approximation problems.
\newblock {\em J. {ACM}}, 45(3):501--555, 1998.

\bibitem{AtseriasK25}
Albert Atserias and Phokion~G. Kolaitis.
\newblock Consistency of relations over monoids.
\newblock {\em J. {ACM}}, 72(3):18:1--18:47, 2025.
\newblock \href {https://doi.org/10.1145/3721855} {\path{doi:10.1145/3721855}}.

\bibitem{DBLP:conf/stoc/BabaiFLS91}
L{\'{a}}szl{\'{o}} Babai, Lance Fortnow, Leonid~A. Levin, and Mario Szegedy.
\newblock Checking computations in polylogarithmic time.
\newblock In {\em {STOC}}, pages 21--31. {ACM}, 1991.

\bibitem{DBLP:journals/corr/abs-2507-18375}
Guillermo Badia, Manfred Droste, Thomas Eiter, Rafael Kiesel, Carles Noguera, and Erik Paul.
\newblock Fagin's theorem for semiring turing machines.
\newblock {\em CoRR}, 2025.
\newblock \href {https://arxiv.org/abs/2507.18375} {\path{arXiv:2507.18375}}.

\bibitem{badia_et_al:LIPIcs.MFCS.2024.14}
Guillermo Badia, Manfred Droste, Carles Noguera, and Erik Paul.
\newblock {Logical Characterizations of Weighted Complexity Classes}.
\newblock In Rastislav Kr\'{a}lovi\v{c} and Anton{\'\i}n Ku\v{c}era, editors, {\em 49th International Symposium on Mathematical Foundations of Computer Science (MFCS 2024)}, volume 306 of {\em Leibniz International Proceedings in Informatics (LIPIcs)}, pages 14:1--14:16, Dagstuhl, Germany, 2024. Schloss Dagstuhl -- Leibniz-Zentrum f{\"u}r Informatik.
\newblock \href {https://doi.org/10.4230/LIPIcs.MFCS.2024.14} {\path{doi:10.4230/LIPIcs.MFCS.2024.14}}.

\bibitem{Jelia25}
Timon Barlag, Nicolas Fr{\"{o}}hlich, Teemu Hankala, Miika Hannula, Minna Hirvonen, Vivian Holzapfel, Juha Kontinen, Arne Meier, and Laura Strieker.
\newblock Logic and computation through the lens of semirings.
\newblock {\em CoRR}, 2025.
\newblock \href {https://arxiv.org/abs/2502.12939} {\path{arXiv:2502.12939}}.

\bibitem{BarlagHKPV23}
Timon Barlag, Miika Hannula, Juha Kontinen, Nina Pardal, and Jonni Virtema.
\newblock Unified foundations of team semantics via semirings.
\newblock In {\em {KR}}, pages 75--85, 2023.

\bibitem{BiziereGN23}
Clotilde Bizi{\`{e}}re, Erich Gr{\"{a}}del, and Matthias Naaf.
\newblock Locality theorems in semiring semantics.
\newblock In {\em {MFCS}}, volume 272 of {\em LIPIcs}, pages 20:1--20:15. Schloss Dagstuhl - Leibniz-Zentrum f{\"{u}}r Informatik, 2023.

\bibitem{DBLP:conf/focs/BlumSS88}
Lenore Blum, Mike Shub, and Steve Smale.
\newblock On a theory of computation over the real numbers; {NP} completeness, recursive functions and universal machines (extended abstract).
\newblock In {\em {FOCS}}, pages 387--397. {IEEE} Computer Society, 1988.

\bibitem{IC2005}
Olivier Bournez, Felipe Cucker, Paulin~Jacob{\'e} de~Naurois, and Jean-Yves Marion.
\newblock Implicit complexity over an arbitrary structure: Quantifier alternations.
\newblock {\em Information and Computation}, 202(2):210--230, 2006.
\newblock \href {https://doi.org/10.1016/j.ic.2005.07.005} {\path{doi:10.1016/j.ic.2005.07.005}}.

\bibitem{BrinkeGM24}
Sophie Brinke, Erich Gr{\"{a}}del, and Lovro Mrkonjic.
\newblock Ehrenfeucht-fra{\"{\i}}ss{\'{e}} games in semiring semantics.
\newblock In {\em {CSL}}, volume 288 of {\em LIPIcs}, pages 19:1--19:22. Schloss Dagstuhl - Leibniz-Zentrum f{\"{u}}r Informatik, 2024.

\bibitem{10.1145/1007352.1007425}
Peter B\"{u}rgisser and Felipe Cucker.
\newblock Counting complexity classes for numeric computations ii: algebraic and semialgebraic sets.
\newblock In {\em Proceedings of the Thirty-Sixth Annual ACM Symposium on Theory of Computing}, STOC '04, page 475–485, New York, NY, USA, 2004. Association for Computing Machinery.
\newblock \href {https://doi.org/10.1145/1007352.1007425} {\path{doi:10.1145/1007352.1007425}}.

\bibitem{10.1145/3651146}
Marco Calautti, Ester Livshits, Andreas Pieris, and Markus Schneider.
\newblock The complexity of why-provenance for datalog queries.
\newblock {\em Proc. ACM Manag. Data}, 2(2), may 2024.
\newblock \href {https://doi.org/10.1145/3651146} {\path{doi:10.1145/3651146}}.

\bibitem{DannertGNT21}
Katrin~M. Dannert, Erich Gr{\"{a}}del, Matthias Naaf, and Val Tannen.
\newblock Semiring provenance for fixed-point logic.
\newblock In Christel Baier and Jean Goubault{-}Larrecq, editors, {\em 29th {EACSL} Annual Conference on Computer Science Logic, {CSL} 2021, January 25-28, 2021, Ljubljana, Slovenia (Virtual Conference)}, volume 183 of {\em LIPIcs}, pages 17:1--17:22. Schloss Dagstuhl - Leibniz-Zentrum f{\"{u}}r Informatik, 2021.
\newblock \href {https://doi.org/10.4230/LIPIcs.CSL.2021.17} {\path{doi:10.4230/LIPIcs.CSL.2021.17}}.

\bibitem{DROSTE200769}
Manfred Droste and Paul Gastin.
\newblock Weighted automata and weighted logics.
\newblock {\em Theoretical Computer Science}, 380(1):69--86, 2007.
\newblock Automata, Languages and Programming.
\newblock \href {https://doi.org/10.1016/j.tcs.2007.02.055} {\path{doi:10.1016/j.tcs.2007.02.055}}.

\bibitem{EiterK23}
Thomas Eiter and Rafael Kiesel.
\newblock Semiring reasoning frameworks in {AI} and their computational complexity.
\newblock {\em J. Artif. Intell. Res.}, 77:207--293, 2023.
\newblock \href {https://doi.org/10.1613/JAIR.1.13970} {\path{doi:10.1613/JAIR.1.13970}}.

\bibitem{eldar_et_al:LIPIcs.ICDT.2024.4}
Idan Eldar, Nofar Carmeli, and Benny Kimelfeld.
\newblock {Direct Access for Answers to Conjunctive Queries with Aggregation}.
\newblock In Graham Cormode and Michael Shekelyan, editors, {\em 27th International Conference on Database Theory (ICDT 2024)}, volume 290 of {\em Leibniz International Proceedings in Informatics (LIPIcs)}, pages 4:1--4:20, Dagstuhl, Germany, 2024. Schloss Dagstuhl -- Leibniz-Zentrum f{\"u}r Informatik.
\newblock \href {https://doi.org/10.4230/LIPIcs.ICDT.2024.4} {\path{doi:10.4230/LIPIcs.ICDT.2024.4}}.

\bibitem{DBLP:journals/ijac/GaubertK06}
Stephane Gaubert and Ricardo Katz.
\newblock Reachability problems for products of matrices in semirings.
\newblock {\em Int. J. Algebra Comput.}, 16(3):603--627, 2006.
\newblock \href {https://doi.org/10.1142/S021819670600313X} {\path{doi:10.1142/S021819670600313X}}.

\bibitem{GradelHNW22}
Erich Gr{\"{a}}del, Hayyan Helal, Matthias Naaf, and Richard Wilke.
\newblock Zero-one laws and almost sure valuations of first-order logic in semiring semantics.
\newblock In {\em {LICS}}, pages 41:1--41:12. {ACM}, 2022.

\bibitem{GM95}
Erich Gr\"{a}del and Klaus Meer.
\newblock Descriptive complexity theory over the real numbers.
\newblock In {\em Proceedings of the Twenty-Seventh Annual ACM Symposium on Theory of Computing}, STOC '95, page 315–324, New York, NY, USA, 1995. Association for Computing Machinery.
\newblock \href {https://doi.org/10.1145/225058.225151} {\path{doi:10.1145/225058.225151}}.

\bibitem{GradelM21}
Erich Gr{\"{a}}del and Lovro Mrkonjic.
\newblock Elementary equivalence versus isomorphism in semiring semantics.
\newblock In {\em {ICALP}}, volume 198 of {\em LIPIcs}, pages 133:1--133:20. Schloss Dagstuhl - Leibniz-Zentrum f{\"{u}}r Informatik, 2021.

\bibitem{abs-1712-01980}
Erich Gr{\"{a}}del and Val Tannen.
\newblock Semiring provenance for first-order model checking.
\newblock {\em CoRR}, abs/1712.01980, 2017.
\newblock URL: \url{http://arxiv.org/abs/1712.01980}.

\bibitem{Grädel2025}
Erich Gr{\"a}del and Val Tannen.
\newblock {\em Provenance Analysis and Semiring Semantics for First-Order Logic}, pages 351--401.
\newblock Springer Nature Switzerland, Cham, 2025.
\newblock \href {https://doi.org/10.1007/978-3-031-86319-6_21} {\path{doi:10.1007/978-3-031-86319-6_21}}.

\bibitem{GreenKT07}
Todd~J. Green, Gregory Karvounarakis, and Val Tannen.
\newblock Provenance semirings.
\newblock In Leonid Libkin, editor, {\em Proceedings of the Twenty-Sixth {ACM} {SIGACT-SIGMOD-SIGART} Symposium on Principles of Database Systems, June 11-13, 2007, Beijing, China}, pages 31--40. {ACM}, 2007.
\newblock \href {https://doi.org/10.1145/1265530.1265535} {\path{doi:10.1145/1265530.1265535}}.

\bibitem{grohe_et_al:LIPIcs.ICDT.2025.9}
Martin Grohe, Christoph Standke, Juno Steegmans, and Jan Van~den Bussche.
\newblock {Query Languages for Neural Networks}.
\newblock In Sudeepa Roy and Ahmet Kara, editors, {\em 28th International Conference on Database Theory (ICDT 2025)}, volume 328 of {\em Leibniz International Proceedings in Informatics (LIPIcs)}, pages 9:1--9:18, Dagstuhl, Germany, 2025. Schloss Dagstuhl -- Leibniz-Zentrum f{\"u}r Informatik.
\newblock \href {https://doi.org/10.4230/LIPIcs.ICDT.2025.9} {\path{doi:10.4230/LIPIcs.ICDT.2025.9}}.

\bibitem{HannulaKBV20}
Miika Hannula, Juha Kontinen, Jan~Van den Bussche, and Jonni Virtema.
\newblock Descriptive complexity of real computation and probabilistic independence logic.
\newblock In {\em {LICS}}, pages 550--563. {ACM}, 2020.

\bibitem{im_et_al:LIPIcs.ICDT.2024.11}
Sungjin Im, Benjamin Moseley, Hung Ngo, and Kirk Pruhs.
\newblock {On the Convergence Rate of Linear Datalog$^\circ$ over Stable Semirings}.
\newblock In Graham Cormode and Michael Shekelyan, editors, {\em 27th International Conference on Database Theory (ICDT 2024)}, volume 290 of {\em Leibniz International Proceedings in Informatics (LIPIcs)}, pages 11:1--11:20, Dagstuhl, Germany, 2024. Schloss Dagstuhl -- Leibniz-Zentrum f{\"u}r Informatik.
\newblock \href {https://doi.org/10.4230/LIPIcs.ICDT.2024.11} {\path{doi:10.4230/LIPIcs.ICDT.2024.11}}.

\bibitem{KOSTOLANY}
Peter Kostolányi.
\newblock Weighted automata and logics meet computational complexity.
\newblock {\em Information and Computation}, 301:105213, 2024.
\newblock \href {https://doi.org/10.1016/j.ic.2024.105213} {\path{doi:10.1016/j.ic.2024.105213}}.

\bibitem{10.5555/639508.639512}
Mehryar Mohri.
\newblock Semiring frameworks and algorithms for shortest-distance problems.
\newblock {\em J. Autom. Lang. Comb.}, 7(3):321–350, January 2002.

\bibitem{munozserrano_et_al:LIPIcs.ICDT.2024.12}
Thomas Mu\~{n}oz Serrano, Cristian Riveros, and Stijn Vansummeren.
\newblock {Enumeration and Updates for Conjunctive Linear Algebra Queries Through Expressibility}.
\newblock In Graham Cormode and Michael Shekelyan, editors, {\em 27th International Conference on Database Theory (ICDT 2024)}, volume 290 of {\em Leibniz International Proceedings in Informatics (LIPIcs)}, pages 12:1--12:20, Dagstuhl, Germany, 2024. Schloss Dagstuhl -- Leibniz-Zentrum f{\"u}r Informatik.
\newblock \href {https://doi.org/10.4230/LIPIcs.ICDT.2024.12} {\path{doi:10.4230/LIPIcs.ICDT.2024.12}}.

\bibitem{Poi95}
Bruno Poizat.
\newblock {\em Les petits cailloux: Une approche mod{\`e}le-th{\'e}orique de l'Algorithmie}.
\newblock Al{\'e}as Editeur, 1995.

\bibitem{PreparataS85}
Franco~P. Preparata and Michael~Ian Shamos.
\newblock {\em Computational Geometry - An Introduction}.
\newblock Texts and Monographs in Computer Science. Springer, 1985.
\newblock \href {https://doi.org/10.1007/978-1-4612-1098-6} {\path{doi:10.1007/978-1-4612-1098-6}}.

\bibitem{DBLP:journals/fuin/RudeanuV04}
Sergiu Rudeanu and Dragos Vaida.
\newblock Semirings in operations research and computer science: More algebra.
\newblock {\em Fundam. Informaticae}, 61(1):61--85, 2004.
\newblock URL: \url{http://content.iospress.com/articles/fundamenta-informaticae/fi61-1-06}.

\bibitem{2407-18006}
Marcus Schaefer, Jean Cardinal, and Tillmann Miltzow.
\newblock The existential theory of the reals as a complexity class: {A} compendium.
\newblock {\em CoRR}, abs/2407.18006, 2024.

\bibitem{SchaeferS17}
Marcus Schaefer and Daniel Stefankovic.
\newblock Fixed points, nash equilibria, and the existential theory of the reals.
\newblock {\em Theory Comput. Syst.}, 60(2):172--193, 2017.
\newblock \href {https://doi.org/10.1007/s00224-015-9662-0} {\path{doi:10.1007/s00224-015-9662-0}}.

\bibitem{DBLP:books/daglib/0086373}
Michael Sipser.
\newblock {\em Introduction to the theory of computation}.
\newblock {PWS} Publishing Company, 1997.

\bibitem{S0097539793243004}
Richard~E. Stearns and Harry~B. Hunt~III.
\newblock An algebraic model for combinatorial problems.
\newblock {\em SIAM Journal on Computing}, 25(2):448--476, 1996.
\newblock \href {https://doi.org/10.1137/S0097539793243004} {\path{doi:10.1137/S0097539793243004}}.

\bibitem{MR248973}
Volker Strassen.
\newblock Gaussian elimination is not optimal.
\newblock {\em Numer. Math.}, 13:354--356, 1969.
\newblock \href {https://doi.org/10.1007/BF02165411} {\path{doi:10.1007/BF02165411}}.

\bibitem{Tannen17}
Val Tannen.
\newblock Provenance analysis for {FOL} model checking.
\newblock {\em {ACM} {SIGLOG} News}, 4(1):24--36, 2017.

\bibitem{tencate_et_al:LIPIcs.ICDT.2024.8}
Balder ten Cate, Victor Dalmau, Phokion~G. Kolaitis, and Wei-Lin Wu.
\newblock {When Do Homomorphism Counts Help in Query Algorithms?}
\newblock In Graham Cormode and Michael Shekelyan, editors, {\em 27th International Conference on Database Theory (ICDT 2024)}, volume 290 of {\em Leibniz International Proceedings in Informatics (LIPIcs)}, pages 8:1--8:20, Dagstuhl, Germany, 2024. Schloss Dagstuhl -- Leibniz-Zentrum f{\"u}r Informatik.
\newblock \href {https://doi.org/10.4230/LIPIcs.ICDT.2024.8} {\path{doi:10.4230/LIPIcs.ICDT.2024.8}}.

\end{thebibliography}

\newpage
\appendix

\section{Proof of Theorem \ref{thm:cook}}\label{app:cook}
\cook*
\begin{proof}
    The proof follows the idea of the proof of Theorem \ref{thm:fagin}.    
    The goal is to build a formula $\varphi$ that describes the computation of a non-deterministic $\BSSK$ machine $M$ on some arbitrary input $x$ such that $\varphi$ is satisfiable if and only if $M$ accepts $x$.
    We use constants from $K$.
    Let $M$ be an arbitrary $\BSSK$ machine that runs in time $T(n)$ for input length $n$ and $\bar{x} = (x_1\dots x_n) \in K^n$ an arbitrary input string. 
    Let $(x_1', \dots, x_m') \in K^m$ be an arbitrary guess of length $m$, where $m$ is polynomially bounded by $n$.
    Similar to the relations in the proof of Fagin's theorem we define the following predicates: 
    \begin{itemize}
        \item $v_{t,p}$ describes the content of the tape at time $t$ and position $p$
        \item $q_{t,s} = \begin{cases}
            1, & \text{if $M$ is at node $s$ at time $t$}\\
            0, & \text{else}
        \end{cases}$
    \end{itemize} 
    The sets of different kinds of nodes of $M$ will be defined in the same way as in the proof of Theorem \ref{thm:fagin} and $i$ denotes the $i$th position.
    Notice that here the position is only used as an index of the predicate for the tape content.
    We use formulae $\varphi^t_v$ to describe the changes of the tape values at time $t$ and formulae $\varphi^t_q$ for the node changes.
    Boolean connectives are used in the same way as in the proof of Theorem~\ref{thm:fagin}.
    
    The formula $\varphi^0_v$ ensures that the input and the guess are encoded on the tape as per Definition \ref{def:nondeterministic_BSS_machine} and $\varphi^0_q$ ensures that the machine is only in the initial node $1$.
    We assume the input $\bar{x} = (x_1, \dots, x_n)$ is given as constants.
    \begin{alignat*}{2}
        \varphi^0_v \coloneqq & \\
        &\smashoperator[l]{\bigwedge_{1 \leq i \leq n}} &&(v_{0,i} = x_i \land v_{0,-i} = 1) \,\land \\
        &\smashoperator[l]{\bigwedge_{n+1 < i \leq T(n)}} &&\bigl((v_{0,-i} = 0 \bto (v_{0, -i - 1} = 0 \land v_{0,i} = 0)) \land (v_{0,-i} = 0 \lor v_{0,-i} = 1) \bigr)\, \land\\
        & &&v_{0,-(n+1)} = 0 \land v_{0,n+1} = 0 \land v_{0,0} = 0
    \end{alignat*}
    \[ \varphi^0_q \coloneqq (q_{0, 1} = 1) \land \bigwedge_{s \in S\setminus\{ 1\}} q_{0,s} = 0 \]
    For each computation step $t \bto t+1$ the formula $\varphi^{t+1}_v$ describes the changes in the value.
    Here we also use $\star_s \in \{-, +\}$ to denote either a right or left shift as in the proof of Theorem~\ref{thm:fagin}. 
    Similarly the positions $i_s, j_s, k_s$ are the ones associated with an arithmetic node $s$, as well as the constant $c_s$.
    \begin{alignat*}{2}
        \varphi^{t+1}_v \coloneqq &\\
        &\smashoperator[l]{\bigwedge_{s \in S_{\text{shift}}}} &&\left(q_{t,s} \bto \bigwedge_{-T(n) \leq p \leq T(n)} v_{t+1, p} = v_{t,p\star_s1}\right)  \land \\
        &\smashoperator[l]{\bigwedge_{s \in S_{\text{branch}}}} &&\left(q_{t,s} \bto \bigwedge_{-T(n) \leq p \leq T(n)} v_{t+1, p} = v_{t,p} \right)\, \land \\
        & \smashoperator[l]{\bigwedge_{s \in S_\text{add}}} &&q_{t,s} \bto (v_{t+1,i_s} = v_{t,j_s} + v_{t,k_s})\, \land\\
        & \smashoperator[l]{\bigwedge_{s \in S_\text{mult}}} &&q_{t,s} \bto (v_{t+1,i_s} = v_{t,j_s} \cdot v_{t,k_s})\, \land \\
        & \smashoperator[l]{\bigwedge_{s \in S_{\text{const}}}} &&q_{t,s} \bto (v_{t+1, i_s} = c_s)
    \end{alignat*}
    Correspondingly we use $\varphi^{t+1}_q$ to describe the changes in the node for one computation step.
    \begin{alignat*}{2}
        \varphi^{t+1}_q \coloneqq &\\
        & \smashoperator[l]{\bigwedge_{s \in S\setminus S_{\text{branch}}}} &&q_{t,s} \bto q_{t+1,\beta(s)}\, \land \\
        & \smashoperator[l]{\bigwedge_{s \in S_\text{branch}}} &&q_{t,s} \bto 
         \bigl((v_{t,1} > v_{t,2} \bto q_{t+1,\beta^+(s)}) \land (v_{t,1} \leq v_{t,2} \bto q_{t+1,\beta^-(s)}) \bigr)
    \end{alignat*}
    
    We require the tape to be of the form $\dots z10\textbf{.} x_1\dots$, where $z \neq 1$ and $x_1 \neq 0$ to be compliant with our definition of output readout of a $\BSSK$ machine. 
    Additionally it enforces that the node and the tape content don't change after the final node $N$ is reached.
    This yields the formula 
    \begin{alignat*}{1}
         \varphi_{\text{out}} \coloneqq& \\
         &\smashoperator[l]{\bigvee_{0 \leq t \leq T(n)}} \left( q_{t, N} \bto (q_{t+1, N} \land v_{t,1} \neq 0 \land v_{t,-1} = 1 \land v_{t,-2}\neq 1) \right) \land\\
         &\smashoperator[l]{\bigwedge_{-T(n) \leq p \leq T(n)}} v_{t,p} = v_{t+1,p}.
    \end{alignat*}
    
    As in the proof of Theorem \ref{thm:fagin} we need a formula $\varphi_\text{constraint}$ that enforces the constraints of our machine model.
    \begin{align*}
    \varphi_\text{constraint} &\coloneqq \bigwedge_{s \in (S_\text{add} \cup S_\text{mult} \cup S_\text{const})} p \neq i \bto (v_{t,p}=v_{t+1,p}) \land \bigwedge_{s \in S} q_{t,s} \bto \left( \bigwedge_{s \neq s' \in S} q_{t,s'} = 0 \right)
    \end{align*}
    The first part corresponds to checking that at most one cell changes for arithmetic nodes (i.\,e. the one at the nodes designated position $i_s$). 
    The second part checks that the machine stays in one node at a single time point. 
    This also implies that there is only one node change per step.

    The satisfiability of the conjunction \[\varphi = \varphi^0_v \land \varphi^0_q \land \bigwedge_{t \in [T(n)]} \left(\varphi^t_v \land \varphi^t_q \right) \land \varphi_\text{out} \land \varphi_\text{constraint}\] 
    now corresponds to the acceptance of our machine $M$.
    The correctness follows in a straightforward way from the construction.
\end{proof}

\section{Proofs for Section \ref{ETKand}}\label{app:etk}

The next lemma is presented in \Cref{ETKand}
in order to show that, for instance, in \Cref{def:npkx}
the assumption about the machine $M$ running in polynomial time (on all inputs in $K^*$)
is no more restrictive than saying that the machine $M$ runs in polynomial time on
inputs in the set $X^*$.
Therefore, we concentrate on main points of the proof of this lemma.

\timer*
\begin{proof}
Let $t$ be a polynomial function that bounds the runtime of the machine $M$.
Using the $K$-Turing machines of \cite{Jelia25}, it is enough to show that there is
a $K$-TM $M'$ such that the input-output function of $M'$ satisfies the condition
$f_{M'} \upharpoonright X^* = f$ and that the machine $M'$ runs in polynomial time for all inputs
in the set $X^*$.
Since $K$-Turing machines can use a finite number of states and a finite alphabet
of tape symbols in similar way to ordinary Turing machines,
the idea of the proof is to proceed in a manner that resembles the use of binary counters
in the context of Turing machines in classical complexity theory.

We give a description of the computation of the machine $M'$. Let $\bar{x} \in K^*$ be
given as the input. First, the symbols of the input are moved such that between each two cells of the tape
of the machine corresponding to this input string, there is a cell filled with the blank symbol $b$.
The machine $M'$ simulates the computation of the $\BSSK$ machine $M$ using only
every other cell of tape, corresponding to how the input is now located.

On the remaining every other cell of the tape, a binary counter is used based on
symbols $0$ and $1$. This is done in the following manner.
The counter is initialized to the value $t(n)$, encoded in binary, where $n$
is the length of the input $\bar{x}$. This can be done in polynomial time,
since the function $t$ is a fixed polynomial function.
After each simulated computation step of the machine $M$, the value of this counter is decremented by one.

If the simulation of the machine $M$ on input $\bar{x}$ halts before the counter reaches the value zero, the counter is discarded and the output of the simulation is arranged on the tape such that it is also the output of the machine $M'$.
However, if the counter hits zero before the simulation has stopped, then the machine $M'$ halts.

In this way, the simulation of the machine $M$ always halts after $t(n)$ computation steps, and for this, the simulating machine $M'$ uses only a polynomial number of computation steps.
Since the function $f$ is computed by $M$ in time $t$, it holds that
$f_{M'} \upharpoonright X^* = f$. On the other hand, the machine $M'$ always halts after simulating at most $t(n)$ computation steps, and in this way
the machine $M'$ halts on all inputs in $K^*$ in polynomial time.
\end{proof}

\nonarithmeticturing*
\begin{proof}
Let $M$ be a polynomial-time non-arithmetic $\BSSK$ machine that computes
the given reduction $Q \leqpmkxx P$.
Let $b \notin X$ be the blank symbol.
We give a description for a Turing machine $M'$ with input alphabet $X$ that runs in
polynomial time and computes the reduction from $Q$ to $P$.
The computation of the machine $M'$ proceeds through three phases.

In the first phase, the input $\bar{x} = (x_1, \dots, x_n) \in X^*$ of the machine $M'$ is converted into
the concatenation $1^n 0 \bar{x}$, that is, into the form
\[
(\underbrace{1, \dots, 1}_{n}, 0, x_1, \ldots, x_n).
\]
This string resembles the way
in which the input $\bar{x}$ is initialized for the $\BSSK$ machine $M$.
In order to produce the correct number of copies of the element $1$, the machine $M'$ may
go through the cells of the string $\bar{x}$, one element at a time,  and use a special
symbol $\hat{a}$ for each $a \in X$ to mark the current position in the string. These special
symbols are then reverted back to the corresponding elements of $X$.
The head of the machine $M'$ is moved to $x_1$ corresponding to the starting position of $M$.

In the second phase, the computation of the non-arithmetic machine $M$ on input $\bar{x}$ is simulated step by step.
Let $Q = \{1, \dots, N\}$ be the set of all computation nodes of the $\BSSK$ machine $M$
as in \Cref{def:BSSK}.
During the simulation, the machine $M'$ may base its computation on the set $Q$ of states
in accordance with the five types of computation nodes of the non-arithmetic machine $M$,
starting from the initial state $1$.
Since the tape is filled with the blank symbol $b$ instead of $0$,
whenever the symbol $b$ is read from the tape, it is treated as if it had the semiring value $0$.
\begin{itemize}
\item For the input node, the state is changed to the next computation node $\beta(1)$ of $M$.
\item Whenever an arithmetic computation node $m \in Q$ is reached, the machine $M'$ proceeds as follows.
If the computation node is a constant assignment for some $c \in X$, this symbol is written on
the tape at the head position of the machine. If the computation node corresponds to an arithmetic
operation of the form $x_i \leftarrow x_j \star x_k$ for some $\star \in \{+, \cdot\}$,
using auxiliary states of the machine,
the elements corresponding to $x_j$ and $x_k$ are retrieved and the result of the operation is written to the tape cell corresponding to $x_i$. Note that the indices $i$, $j$ and $k$ are fixed.
The result of the arithmetic operation is the value of the elements $x_j$ and $x_k$
which is not neutral element of $\star$,
or if both of them are equal to the neutral element, then this element is the result of the operation.
The state of the machine is changed to $\beta(m)$.
\item For a branching node $m \in Q$ of $M$, the possible results of the comparisons `$=$' or `$\leq$'
between any two elements of the finite set $X$ are encoded in auxiliary states of the Turing machine $M'$.
The state of the machine corresponding to the set $Q$ is changed
to $\beta^+(m)$ or $\beta^-(m)$ based on the result of the comparison.
\item For a shift node $m$, the tape is shifted accordingly
and the state is changed to $\beta(m)$.
\item Reaching the output node $N$ of $M$ concludes the simulation phase.
\end{itemize}

After the simulation, in the third and final phase of the computation of $M'$, the content of
the working tape of $M'$ is of the form
\[
    (\ldots, z, \underbrace{1, \dots, 1}_\ell, 0 \textbf{.} y_1, \ldots, y_\ell, \ldots )
\]
for some $z \neq 1$, where the machine head is at the marked position.
From this result, the string $(y_1, \dots, y_\ell)$ is extracted to be the output of $M'$.

When defined in this way, the Turing machine $M'$ runs in polynomial time and computes
the reduction $Q \leqpm P$ corresponding to the same reduction $Q \leqpmkxx P$ that is computed by
the $\BSSK$ machine $M$.
\end{proof}

\end{document}